\documentclass[12pt]{article}
\usepackage{amsmath}
\usepackage{amssymb}
\usepackage{amsthm}
\usepackage{graphicx,psfrag,epsf}
\usepackage{psfrag,epsf}
\usepackage{natbib}
\usepackage{url} 
\usepackage[titletoc,title]{appendix}
\usepackage{array}
\usepackage{bm}
\usepackage{color}
\usepackage{graphicx}
\usepackage{epstopdf}
\usepackage[letterpaper, portrait, margin=1in]{geometry}
\usepackage{lipsum}
\usepackage{listings}
\usepackage{fancyref}
\usepackage{float}
\usepackage[font=footnotesize,labelfont=bf]{caption}
\usepackage{color}
\usepackage{lscape}
\usepackage{chngcntr}
\usepackage[shortlabels]{enumitem}

\usepackage{pdfpages}

\newtheorem{theorem}{Theorem}[section]
\newtheorem{lemma}{Lemma}[subsection]

\newtheorem{corollary}[lemma]{Corollary}

\newcommand{\wT}{w^1}
\newcommand{\wC}{w^0}
\newcommand{\wTbar}{\overline{w}^1}
\newcommand{\wCbar}{\overline{w}^0}
\newcommand{\wtbar}{\overline{w}^t}
\newcommand{\wbar}{\overline{w}}
\newcommand{\barbarY}{\bar{\bar{Y}}}
\newcommand{\barbary}{\bar{\bar{y}}}
\newcommand{\xbarT}{\xbarbar^1}
\newcommand{\xbarC}{\xbarbar^0}
\newcommand{\xbarbar}{\overline{\overline{\bm{x}}}}

\theoremstyle{remark}
\newtheorem*{remark}{Remark}

\title{Supplementary Materials for ``Design-Based Ratio Estimators and Central Limit Theorems for Clustered, Blocked RCTs''}
\date{}
\begin{document}
\includepdf[pages=1-last]{"Design_Based_Ratio_Estimators".pdf}

\appendix
\setcounter{page}{1}
\counterwithin{table}{section}
\maketitle

\section{Proofs for asymptotic results}\label{sec:main_tools}

\subsection{Asymptotic properties of ratio estimators}\label{sec:main_tools}
In this section, we develop two results that will help us prove asymptotic normality of our weighted least squares regression estimators.
We derive these results generally, assuming arbitrary cluster level\footnote{The same mathematical arguments are immediately extendable to the case where we have individuals rather than clusters.} potential outcomes $C_j(t)$ for cluster $j$ with $t \in \{0, 1\}$.
These results could alternatively be obtained using a finite population delta method \citep[see][]{pashley2019delta}.

\subsubsection{Asymptotic normality of one ratio}\label{subsec:asy_one_ratio_main}
We first derive a result on the asymptotic distribution of a single ratio estimator, very similar to Theorem 1 in \citet{scott1981asymptotic}, but using conditions from \citet{li2017general}.
To do this, we redefine the potential outcomes in a way that allows us to use known normality results for an estimator that is close to our ratio estimator.
Then we use an application of Slutsky's theorem to get the normality result for our estimator of interest.
Before the result, we must define some notation.
Let us have arbitrary cluster level potential outcomes $C_j(t)$ for $t \in \{0, 1\}$ and cluster weights $w_j$, with $T_j$ being the indicator of treatment assignment for cluster $j$.
Following \citet{li2017general}, our finite population is within a sequence of finite populations where $m$ and $m^t$ go to $\infty$, which naturally holds if we grow the population with a fixed proportion of treated clusters, $p$.
Further define
\begin{align*}
\wbar &= \frac{1}{m}\sum_{j = 1}^{m}w_{j},\\
\wtbar &= \frac{1}{m^t}\sum_{j: T_j = t}w_{j},\\
\overline{\overline{C}}(t) &= \frac{\frac{1}{m}\sum_{j = 1}^{m}w_jC_j(t)}{\wbar} \quad \text{(the overall weighted mean of $C_j(t)$)},\\
\overline{\overline{c}}(t) &= \frac{\frac{1}{m^t}\sum_{j: T_j=t}w_jC_j(t)}{\wtbar},
\end{align*}
and
\begin{align*}
\hat{\bar{C}}(t) &= \overline{\overline{c}}(t)\wbar.
\end{align*}

In particular, we are interested in the asymptotic distribution of the ratio estimator $\overline{\overline{c}}(t)$, unlike \citet{scott1981asymptotic} who found the distribution for an estimator closer to the form $\hat{\bar{C}}(t)$.
In our case, $\overline{\overline{c}}(t)$ is an estimator of $\overline{\overline{C}}(t)$.
We will see this is a straightforward change.

Define the following finite population variance of the weights:
\begin{align*}
S^2(w) &= \frac{1}{m-1}\sum_{j=1}^{m}(w_j - \wbar)^2.
\end{align*}
Let $D_j(t)$ by a scaled deviation of $C_j(t)$ from the overall weighted mean.
That is,
\[D_{j}(t) = \left(w_jC_j(t) - w_{j}\overline{\overline{C}}(t)\right)/\wbar.\]
Then the mean of the $D_{j}(t)$'s, $\bar{D}(t)$, equals 0 as 
\[\sum_{j=1}^{m}w_j\overline{\overline{C}}(t) = m\wbar \overline{\overline{C}}(t) = \sum_{j=1}^{m}w_jC_j(t),\]
and the variance expression for the $D_j(t)$, which is a weighted variance of the $C_j(t)$, is
\[S^2_D(t) = \frac{1}{m-1}\sum_{j=1}^{m}D_{j}^2(t) = \frac{1}{m-1}\sum_{j=1}^{m}\frac{w_{j}^2}{\wbar^2}\left(C_j(t)-\overline{\overline{C}}(t)\right)^2.\]

\begin{lemma}\label{lem:main_res}
Assume we have the following conditions:
\begin{enumerate}[(a)]
\item Defining $g(t)= \text{max}_{1 \leq j \leq m}(D_{j}(t))^2$
\begin{align}
\frac{1}{\text{min}(m^1, m-m^1)}\frac{g(t)}{S^2_{D}(t)} \to 0, \quad \text{as $m \to \infty$.} \label{eq:var_cond_g}
\end{align}
\item For $t \in \{1, 0\}$,
\begin{align}
\left(1-f^t\right)\frac{S^2(w)}{m^t\wbar^2} \to 0 \text{ as } m \to \infty. \label{eq:var_cond}
\end{align}
\end{enumerate}
Under these conditions,
\begin{align}
 \frac{\overline{\overline{c}}(t)- \overline{\overline{C}}(t)}{\sqrt{\frac{1-f^t}{m^t}}S_{D}(t)/\wbar}
& \xrightarrow[]{d} N(0,1). \label{eq:result_one_ratio}
\end{align}
\end{lemma}

\begin{proof}
The proof of this result closely follows the proof in \citet{scott1981asymptotic}. 
We are interested in the asymptotic distribution of the ratio estimator $\overline{\overline{c}}(t)$, but to find this we will first find the distribution of $\hat{\bar{C}}(t)$.

Let $\bar{d}(t) = \sum_{j : T_j = t}D_{j}(t)/m^t$.
Then
\begin{align*}
\bar{d}(t) 
&=\left(\overline{\overline{c}}(t) - \overline{\overline{C}}(t)\right)\frac{\wbar^t}{\wbar}.
\end{align*}
Assuming our first condition given by Equation~\ref{eq:var_cond_g}, by Theorem 1 of \citet{li2017general},
\[\frac{\bar{d}(t)}{\sqrt{\frac{1-f^t}{m^t}}S_{D}(t)} \xrightarrow[]{d} N(0,1).\]
Further, from Theorem B of \citet{scott1981asymptotic}, given our second condition given by Equation~\ref{eq:var_cond},
\[\wbar^t/\wbar - 1 \xrightarrow[]{p}  0 \text{ as } m \to \infty.\]
Thus, $\wbar^t/\wbar \xrightarrow[]{p} 1$.
Note the condition on the weights would cause issues if the average weight were to go to 0 as the sample size increased, but, as discussed in Section~\ref{subsec:cons_no_cov}, this is easy to avoid by using unnormalized weights.
This in turn means by Slutsky's Theorem, recalling that $\overline{\overline{c}}(t) - \overline{\overline{C}}(t) = \bar{d}(t)\wbar/\wbar^t$, we have our result in Equation~\ref{eq:result_one_ratio}. 


\end{proof}

\subsubsection{Asymptotic normality of difference of ratios (Main result)}\label{subsec:asy_two_ratios_main}
We will now show the asymptotic normality of a difference of ratios.
In particular, we are interested in the asymptotic distribution of $\overline{\overline{c}}(1)- \overline{\overline{c}}(0) $, using notation from Section~\ref{subsec:asy_one_ratio_main}.
We will again define new potential outcomes that are easier to work with than the direct ratio.
First we define some notation.
Construct new potential outcomes $D_j(t) = (w_jC_j(t)-w_j\overline{\overline{C}}(t))/\wbar$ with notation defined in Section~\ref{subsec:asy_one_ratio_main}, so $\bar{D}(1) = 0$, $\bar{D}(0) = 0$ and
\begin{align*}
\bar{d}(t) &= \frac{1}{m^t}\sum_{j: T_j =t}D_j(t).
\end{align*}
Further define an intermediate estimator $\hat{D} = \bar{d}(1) - \bar{d}(0)$ and
\[S^2_D(1,0) = \frac{1}{m-1}\sum_{j=1}^mD_j(1)D_j(0).\]
Recall for $t \in \{0,1\}$
\[S^2_D(t) = \frac{1}{m-1}\sum_{j=1}^mD_j^2(t).\]
and further define
\[S^2(D) = \frac{1}{m-1}\sum_{j=1}^m(D_j(1) - D_j(0))^2.\]
We will use these notations to get the final result.


\begin{lemma}\label{lemma:two_ratio_main}
Let us have the following conditions\footnote{It is interesting to note that a delta method approach would rely on slightly different conditions. In particular, the delta method does not explicitly require limiting values on asymptotic variances, as given in the second condition, but it does require a convergence in probability result that is satisfied by those limiting values for the asymptotic variances. Comparing this derivation to proofs via the delta method is left for future work.}:
\begin{enumerate}[(a)]
\item Let $g(t) = \text{max}_{1 \leq j \leq m}(D_j(t))^2$ and as $m \to \infty$ 
\begin{align}
\text{max}_{t \in \{0,1\}} \frac{1}{(m^t)^2}\frac{g(t)}{\text{Var}(\hat{D})} \to 0.\label{cond:two_ratio_1}
\end{align}
\item $m^t/m$ has a limiting value in (0,1) and $S^2_D(t)$ and $S^2_D(1,0)$ have limiting values.
\item For $t \in \{1, 0\}$,
\begin{align}
\left(1-f^t\right)\frac{S^2(w)}{m^t\wbar^2} \to 0 \text{ as } m \to \infty. \label{eq:var_cond2}
\end{align}
\end{enumerate}
Then we have that
\begin{align}
\frac{\left(\overline{\overline{c}}(1)- \overline{\overline{c}}(0) \right) - \left(\overline{\overline{C}}(1) - \overline{\overline{C}}(0)\right)}{\sqrt{\text{Var}(\hat{D})}}  \xrightarrow[]{d} N(0, 1) \label{eq:main_res}
\end{align}
where
\[\text{Var}(\hat{D}) = \frac{S^2_D(1)}{m^1} + \frac{S^2_D(0)}{m^0} - \frac{S^2(D)}{m}.\]
\end{lemma}

\begin{proof}
First we derive a normality result for $\hat{D}$.
If we randomly select $pm$ clusters to assign to treatment, then from Theorem 3 in \citet{li2017general}, $\hat{D}$ has mean $0$ and variance
\[\text{Var}(\hat{D}) = \frac{S^2_D(1)}{m^1} + \frac{S^2_D(0)}{m^0} - \frac{S^2(D)}{m}.\]
If we have the condition of Equation~\ref{cond:two_ratio_1}, then under Theorem 4 of \citet{li2017general},
\[\frac{\hat{D} }{\sqrt{\text{Var}(\hat{D})}} \xrightarrow[]{d} \text{N}(0,1).\]
Now note that
\begin{align*}
\frac{\hat{D}}{\sqrt{\text{Var}(\hat{D})}} &= \frac{\bar{d}(1) -\bar{d}(0)}{\sqrt{\text{Var}(\hat{D})}}\\
&= \frac{\frac{1}{m^1}\sum_{j=1}^mT_j\left(w_jC_j(1)-w_j\overline{\overline{C}}(1)\right)/\wbar- \frac{1}{m^0}\sum_{j=1}^m(1-T_j)\left(w_jC_j(0)-w_j\overline{\overline{C}}(0)\right)/\wbar}{\sqrt{\text{Var}(\hat{D})}}\\
&= \frac{\frac{\wTbar}{\wbar}\left(\overline{\overline{c}}(1)-\overline{\overline{C}}(1)\right)-\frac{\wCbar}{\wbar}\left(\overline{\overline{c}}(0)-\overline{\overline{C}}(0)\right)}{\sqrt{\text{Var}(\hat{D})}}.
\end{align*}

We can now get back to our estimator of interest, $\overline{\overline{c}}(1)- \overline{\overline{c}}(0)$:
\begin{align*}
 &\frac{\left(\overline{\overline{c}}(1)-\overline{\overline{C}}(1)\right)-\left(\overline{\overline{c}}(0)-\overline{\overline{C}}(0)\right)}{\sqrt{\text{Var}(\hat{D})}}\\
 &= \frac{\left(1-\frac{\wTbar}{\wbar}\right)\left(\overline{\overline{c}}(1)-\overline{\overline{C}}(1)\right)-\left(1-\frac{\wCbar}{\wbar}\right)\left(\overline{\overline{c}}(0)-\overline{\overline{C}}(0)\right)}{\sqrt{\text{Var}(\hat{D})}}\\
 &\quad +  \frac{\frac{\wTbar}{\wbar}\left(\overline{\overline{c}}(1)-\overline{\overline{C}}(1)\right)-\frac{\wCbar}{\wbar}\left(\overline{\overline{c}}(0)-\overline{\overline{C}}(0)\right)}{\sqrt{\text{Var}(\hat{D})}}\\
  &= \left(1-\frac{\wTbar}{\wbar}\right)\frac{\left(\overline{\overline{c}}(1)-\overline{\overline{C}}(1)\right)}{\sqrt{\text{Var}(\hat{D})}} - \left(1-\frac{\wCbar}{\wbar}\right)\frac{\left(\overline{\overline{c}}(0)-\overline{\overline{C}}(0)\right)}{\sqrt{\text{Var}(\hat{D})}}\\
 &\quad + \frac{\hat{D}}{\sqrt{\text{Var}(\hat{D})}}.
\end{align*}

We know the asymptotic distribution of the third term, so we need to show that the first two terms vanish.
As stated in Lemma~\ref{lem:main_res}, if we have the conditions given by Equation~\ref{cond:two_ratio_1}
and Equation~\ref{eq:var_cond2}, then 
\begin{align*} 
\frac{\overline{\overline{c}}(t)-\overline{\overline{C}}(t)}{\sqrt{\frac{1-f^t}{m^t}}S_{D}(t)}  \xrightarrow[]{d} N(0,1).
\end{align*}
Recall from Section~\ref{subsec:asy_one_ratio_main} that the fourth condition gives us 
\[\frac{\wbar^t}{\wbar} \xrightarrow[]{p} 1.\]

To use these notes, we need to rewrite out previous result:
\begin{align*}
 \frac{\left(\overline{\overline{c}}(1)-\overline{\overline{C}}(1)\right)-\left(\overline{\overline{c}}(0)-\overline{\overline{C}}(0)\right)}{\sqrt{\text{Var}(\hat{D})}}=& \left(1-\frac{\wTbar}{\wbar}\right)\frac{\sqrt{\frac{1}{m^1}-\frac{1}{m}}S_{D}(1)}{\sqrt{\text{Var}(\hat{D})}}\frac{\left(\overline{\overline{c}}(1)-\overline{\overline{C}}(1)\right)}{\sqrt{\frac{1}{m^1}-\frac{1}{m}}S_{D}(1)}\\
 & - \left(1-\frac{\wCbar}{\wbar}\right)\frac{\sqrt{\frac{1}{m^0}-\frac{1}{m}}S_{D}(0)}{\sqrt{\text{Var}(\hat{D})}}\frac{\left(\overline{\overline{c}}(0)-\overline{\overline{C}}(0)\right)}{\sqrt{\frac{1}{m^0}-\frac{1}{m}}S_{D}(0)}\\
 &+ \frac{\hat{D}}{\sqrt{\text{Var}(\hat{D})}}.
 \end{align*}
The second condition puts limiting values on our asymptotic variances.
That is, the second condition (see Theorems 3 and 5 of \citet{li2017general}), gives us that for $t \in \{0,1\}$, $m\left(\frac{1-f^t}{m^t}\right)S^2_{D}(t)$ has some limiting value $\sigma^2_D(t)$ and $m\text{Var}(\hat{D})$ has some limiting value $V$.
This allows us to use Slutsky's theorem to get
\begin{align*}
\left(1-\frac{\wTbar}{\wbar}\right)\frac{\sqrt{\frac{1}{m^1}-\frac{1}{m}}S_{D}(t)/\wbar}{\sqrt{\text{Var}(\hat{D})}}\frac{\left(\overline{\overline{c}}(t)-\overline{\overline{C}}(t)\right)}{\sqrt{\frac{1-f^t}{m^t}}S_{D}(t)} \xrightarrow[]{p} 0.
\end{align*}
To break this step down more, note that the first term goes to 0 in probability, the second goes to a constant by the third condition, and the final term is asymptotically normal.

Finally, Slutsky's theorem gives the result in Equation~\ref{eq:main_res}.
\end{proof}


\subsection{Theorem 1 without covariates}\label{sec:theor_1_no_cov_1_blk}
In this section, we explore the impact estimator of a weighted least squares regression without covariate adjustment.
We prove Theorem 1 under this special case and also find the bias of our estimator.
We start by deriving the form of the estimator, then finding bias and consistency results.
In particular, the estimator has finite sample bias but is consistent.
We then show the asymptotic normality result given in Theorem 1.
Without covariates, the interactions in our regression imply that block means and effects are estimated independently, and so we can focus on the estimator for a single block, $b$.
In other words, for now we can act as if there is only one block.

\subsubsection{The estimator}
We assume that we have one block, arbitrarily block $b$, with $m_b$ clusters.
We assign $m_b^1 = p_bm_b$ clusters to treatment and the rest of the $m_b^0$ clusters to control.
Define shorthand $f_b^t = m_b^t/m_b$.
When we have multiple blocks, we have block indicators $S_{ijb}$ and $S_{jb}$ which equals 1 if unit $i$ or cluster $j$ belong to block $b$, and 0 otherwise.
Each unit $i$, of $n_{jb}$ units, in cluster $j$ of block $b$ has associated weight $w_{ijb}$ and we have $w_{jb} = \sum w_{ijb}$.
Denote $\wC_b=\sum_{j=1}^{m_b}(1-T_{jb})w_{jb}$, $\wT_b=\sum_{j=1}^{m_b}T_{jb}w_{jb}$, and $w_b=\sum_{j=1}^{m_b}w_{jb}$.
Each unit $i$ in cluster $j$ has potential outcomes $Y_{ijb}(1)$ under treatment and $Y_{ijb}(0)$ under control, with treatment assigned at the cluster level.
$y_{ijb}$ is the observed outcome for unit $i$ in cluster $j$ of block $b$ and $\bar{y}_{jb}$ is the observed weighted average outcome for all units in cluster $j$ of block $b$.
We have by definition that $y_{ijb} = T_{jb}Y_{ijb}(1) + (1-T_{jb})Y_{ijb}(0)$ and for $t \in \{0,1\}$,
\[\bar{Y}_{jb}(t) = \frac{1}{w_{jb}}\sum_{i: T_{jb}=t}w_{ijb}Y_{ijb}(t) \]
 and $\bar{y}_{jb} =T_{jb}\bar{Y}_{jb}(1) + (1-T_{jb})\bar{Y}_{jb}(0)$.
We also define for $t \in \{0,1\}$,
\begin{align*}
&\barbary_b(t) = \frac{1}{\sum_{j:T_{jb} = t}w_{jb}}\sum_{j:T_{jb} = t}w_{jb}\bar{y}_{jb}.
\end{align*}
We have finite population parameters
\[\barbarY_b(t)= \frac{1}{\sum_{j=1}^{m_b}w_{jb}}\sum_{j=1}^{m_b}w_{jb}\bar{Y}_{jb}(t).\]

Throughout, we will treat $\beta_{1,b}$ as the \textit{finite population} parameter of the weighted average of cluster impacts, 
\[\beta_{1,b} = \frac{\sum_{j=1}^{m_b}w_{jb}(\bar{Y}_{jb}(1)-\bar{Y}_{jb}(0))}{\sum_{j=1}^{m_b}w_{jb}} = \barbarY_b(1) - \barbarY_b(0).\]

We start by finding the exact form of the regression estimator $\hat{\beta}_1$.
For our regression, we have $\bm{z}_{ijb} = (S_{ijb}\tilde{T}_{jb}, S_{ijb} ) = (\tilde{T}_{jb}, 1)$ because we only have one block.
Here we define
\[p_b^* =  \frac{1}{w_b}\sum_{j=1}^{m_b}T_{jb}w_{jb} = \frac{\wT_b}{w_b}\]
and
\[\tilde{T}_{jb} = T_{jb} - p_b^*,\]
so that
\[\sum_{j=1}^{m_b}w_{j_b}\tilde{T}_{jb} = \sum_{j=1}^{m_b}w_{j_b}\left(T_{jb} - p_b^*\right)=0 .\]
The estimated parameter vector from weighted least squares regressing the individual $y_{ijb}$ on an intercept and centered treatment indicator $\tilde{T}_{jb}$ with weights $w_{ijb}$ is
\begin{align*}
\begin{pmatrix}
\hat{\beta}_{1,b}\\
\hat{\beta}_{0,b}
\end{pmatrix} = 
\left[\left(\sum_{j=1}^{m_b}\sum_{i=1}^{n_{jb}}w_{ijb}\bm{z}'_{ijb}\bm{z}_{ijb}\right)^{-1}\sum_{j=1}^{m_b}\sum_{i=1}^{n_{jb}}w_{ijb}\bm{z}'_{ijb}y_{ijb}\right].
\end{align*}

\noindent\textbf{Result:} $\hat{\beta}_{1,b} =\barbary_b(1) - \barbary_b(0)$.

We have by simple algebra simplifications and matrix inversion,
\begin{align*}
\begin{pmatrix}
\hat{\beta}_{1,b}\\
\hat{\beta}_{0,b}
\end{pmatrix} &=
\begin{pmatrix}
\sum_{j=1}^{m_b}w_{jb}\tilde{T}_{jb}^2 &  \sum_{j=1}^{m_b}w_{jb}\tilde{T}_{jb}\\
\sum_{j=1}^{m_b}w_{jb}\tilde{T}_{jb}  & \sum_{j=1}^{m_b}w_{jb}
\end{pmatrix}^{-1}
\begin{pmatrix}
\sum_{j=1}^{m_b}w_{jb}\tilde{T}_{jb}\bar{y}_{jb}\\
\sum_{j=1}^{m_b}w_{jb}\bar{y}_{jb}
\end{pmatrix}\\
&=
\begin{pmatrix}
\sum_{j=1}^{m_b}w_{jb}\tilde{T}_{jb}^2 &  0\\
0  & w_b
\end{pmatrix}^{-1}
\begin{pmatrix}
\sum_{j=1}^{m_b}w_{jb}\tilde{T}_{jb}\bar{y}_{jb}\\
\wT_b\barbary_b(1) + \wC_b\barbary_b(0)
\end{pmatrix}\\
&=
\begin{pmatrix}
\frac{1}{\sum_{j=1}^{m_b}w_{jb}\tilde{T}_{jb}^2} &  0\\
0  & \frac{1}{w_b}
\end{pmatrix}
\begin{pmatrix}
\sum_{j=1}^{m_b}w_{jb}\tilde{T}_{jb}\bar{y}_{jb}\\
\wT_b\barbary_b(1) + \wC_b\barbary_b(0)
\end{pmatrix}
\end{align*}
To simplify this final form, note a few useful algebra results:
\begin{align*}
\sum_{j=1}^{m_b}w_{jb}\tilde{T}_{jb}^2 &= \sum_{j=1}^{m_b}w_{jb}\left(T_{jb}^2-2T_{jb}p_b^*+(p_b^*)^2\right)\\
&=\wT_b-\frac{(\wT_b)^2}{w_b}\\
&=\frac{\wT_b\wC_b}{w_b}
\end{align*}
and
\begin{align*}
\sum_{j=1}^{m_b}w_{jb}\tilde{T}_{jb}\bar{y}_{jb} &= \sum_{j=1}^{m_b}w_{jb}(T_{jb}-p_b^*)\bar{y}_{jb}\\
&= \wT_b\barbary_b(1) - \frac{\wT_b}{w_b}(\wT_b\barbary_b(1) + \wC_b\barbary_b(0))\\
&= \frac{\wT_b\wC_b}{w_b}\left(\barbary_b(1) - \barbary_b(0)\right).
\end{align*}

We then have the following for our regression estimator for the treatment impact (and control mean):
\begin{align*}
\begin{pmatrix}
\hat{\beta}_{1,b}\\
\hat{\beta}_{0,b}
\end{pmatrix} 
&= \begin{pmatrix}
\frac{w_b}{\wT_b\wC_b} &  0\\
0  & \frac{1}{w_b}
\end{pmatrix}
\begin{pmatrix}
\frac{\wT_b\wC_b}{w_b}\left(\barbary_b(1) - \barbary_b(0)\right)\\
\wT_b\barbary_b(1) + \wC_b\barbary_b(0)
\end{pmatrix}\\
&=\begin{pmatrix}
\barbary_b(1) - \barbary_b(0)\\
\frac{\wT_b}{w_b}\barbary_b(1) + \frac{\wC_b}{w_b}\barbary_b(0)
\end{pmatrix}.
\end{align*}

\subsubsection{Finite sample bias}\label{subsec:bias}
Here we explore the finite sample bias of the regression estimator.
Using Theorem 1 in \cite{middleton2008bias} or Equation 7 in \citet{middleton2015unbiased}, originally from \citet{hartley1954unbiased}, the bias for a ratio estimator is
\begin{align*}
E\left[\frac{u}{v}\right] = \frac{1}{E[v]}\left[E[u]-\text{Cov}\left(\frac{u}{v}, v\right)\right],
\end{align*}
assuming that $v>0$.

We then have, assuming that $\wT>0$ and $\wC>0$,
\begin{align*}
E[\barbary_b(1)]&=E\left[\frac{\sum_{j=1}^{m_b}T_{jb}w_{jb}\bar{y}_{bj}}{\sum_{j=1}^{m_b}w_{jb}T_{jb}}\right]\\
& = \frac{1}{\frac{p_b}{m^1_b}\sum_{j=1}^{m_b}w_{jb}}\left[\frac{p_b}{m^1_b}\sum_{j=1}^{m_b}w_{jb}\bar{Y}_{jb}(1)-\text{Cov}\left(\frac{\sum_{j=1}^{m_b}T_{jb}w_{jb}\bar{y}_{jb}}{\sum_{j=1}^{m_b}w_{jb}T_{jb}}, \frac{1}{m^1_b}\sum_{j=1}^{m_b}T_{jb}w_{jb}\right)\right]\\
& = \barbarY_b(1)-\frac{m_b}{w_b}\text{Cov}\left(\barbary_b(1), \frac{\wT_b}{m^1_b}\right).
\end{align*}

It is then straightforward to find the bias of our estimator as follows:
\begin{align*}
E[\hat{\beta}_{1,b} ] - \beta_{1,b}&= -\frac{m_b}{w_b}\text{Cov}\left(\barbary_b(1), \frac{\wT_b}{m_b^1}\right) + \frac{m_b}{w_b}\text{Cov}\left(\barbary_b(0), \frac{\wC_b}{m_b^0}\right).
\end{align*}
This means, as pointed out by \citet{middleton2015unbiased}, that the bias will depend on (1) the covariance between cluster size (as captured by the $w_{jb}$) and outcomes and (2) the variability of cluster sizes.

\subsubsection{Consistency}\label{subsec:cons_no_cov}

Here we give conditions for consistency of the regression estimator.
It is useful to assume limiting values on the finite population parameters.
This can be handled in a few ways.
Following \citet{middleton2015unbiased}, we may envision a sequence of $h$ finite populations such that as $h \to \infty$, the finite population increases by copying the original $m_b$ clusters in block $b$ $h$ times and then randomization occurs independently within each copy with probability $p_b$.
This keeps our estimands and parameters constant as the population grows.

A bit more generally, we may have limiting values for the finite population parameters.
For a given finite population with $m$ clusters, define the following quantities:
\[\overline{wY(1)}_b = \frac{1}{m_b}\sum_{j=1}^{m_b}w_{jb}\bar{Y}_{jb}(1),\]
\[\overline{wY(0)}_b = \frac{1}{m_b}\sum_{j=1}^{m_b}w_{jb}\bar{Y}_{jb}(0),\] 
and
\[\wbar_b = \frac{1}{m_b}\sum_{j=1}^{m_b}w_{jb}.\]
Given these definitions, our $\beta_{1,b}$ is then
\begin{align*}
\beta_{1,b} & = \frac{\overline{wY(1)}_b}{\wbar_b} - \frac{\overline{wY(0)}_b}{\wbar_b}.
\end{align*}
Assume limiting values  $\mu^*_{b}(1)$, $\mu^*_{b}(0)$, and $\omega_b$ such that as $m_b \to \infty$,
\begin{align*}
\overline{wY(1)}_b \xrightarrow[]{p} \mu^*_{b}(1),\\
\overline{wY(0)}_b \xrightarrow[]{p} \mu^*_{b}(0),
\end{align*}
and
\begin{align*}
\wbar_b &\xrightarrow[]{p} \omega_b >0.
\end{align*}
If we were normalizing the weights across the whole population, it could be the case that $\omega_b \to 0$; we eliminate this issue by using unnormalized weights (which don't change the estimand or estimator because we use weights in the numerator and denominator).
This gives a limiting value for $\beta_{1,b}$ of $\beta_{1,b}^* = (\mu^*_b(1) -\mu^*_b(0))/\omega_b$.

We next show that $\hat{\beta}_{1,b} \xrightarrow[]{p} \beta_{1,b}^*$.
Denote
\[S^2_{w,b}(1) = \frac{1}{m_b-1}\sum_{j=1}^{m_b}\left(w_{jb}\bar{Y}_{jb}(1)-\overline{wY(1)}_b\right)^2,\]
\[S^2_{w,b}(0) = \frac{1}{m_b-1}\sum_{j=1}^{m_b}\left(w_j\bar{Y}_{jb}(0)-\overline{wY(0)}_b\right)^2,\]
and
\[S^2(w_b) = \frac{1}{m_b-1}\sum_{j=1}^{m_b}\left(w_{jb}-\wbar_b\right)^2.\]
Assume that $m_b^1/m_b \to p_b \in (0,1)$.
Using Theorem B from \citet{scott1981asymptotic}, under simple random sampling of clusters into treatment, if our weighted variances do not go to infinity as we increase sample size, specifically if as $m_b \to \infty$,
\begin{align*}
S^2_{w,b}(1)/m_b \to 0,\\
S^2_{w,b}(0)/m_b \to 0,
\end{align*}
and
\begin{align*}
S^2(w_b)/m_b \to 0
\end{align*}
then
\[\frac{1}{m_b^1}\sum_{j=1}^{m_b}w_{jb}T_{jb}\bar{Y}_{jb}(1) - \overline{wY(1)}_b\xrightarrow[]{p} 0,\]
\[\frac{1}{m_b^0}\sum_{j=1}^{m_b}w_{jb}(1-T_{jb})\bar{Y}_{jb}(0) - \overline{wY(0)}_b \xrightarrow[]{p} 0,\]
\[\frac{1}{m_b^1}\sum_{j=1}^{m_b}w_{jb}T_{jb} - \wbar_b \xrightarrow[]{p} 0,\]
and
\begin{align*}
\frac{1}{m_b^0}\sum_{j=1}^{m_b}w_{jb}(1-T_{jb}) - \wbar_b \xrightarrow[]{p} 0
\end{align*}
as $m_b \to \infty$.


We next need a small convergence lemma:
\begin{lemma}\label{lemma:simple_conv_res}
If $A_m -B_m \xrightarrow[]{p} 0$ and $B_m \xrightarrow[]{p} k$, with $k$ a constant, as $m \to \infty$ then $A_m \xrightarrow[]{p} k$ as $m \to \infty$.
\end{lemma}
\begin{proof}
\[A_m - k = A_m -B_m + B_m -k \xrightarrow[]{p} 0.\]
\end{proof}


Using Lemma~\ref{lemma:simple_conv_res}, we have
\[\frac{1}{m_b^t}\sum_{j: T_{jb} = t}w_{jb}\bar{Y}_{jb}(t) = \frac{w_b^t}{m_b^t}\barbary_b(t)  \xrightarrow[]{p} \mu^*_b(t)\]
and
\[\frac{1}{m_b^t}\sum_{j: T_{jb} = t}w_{jb} = \frac{w_b^t}{m_b^t}  \xrightarrow[]{p} \omega_b.\]


Then by Slutsky's theorem,
\begin{align*}
\hat{\beta}_{1,b} = \barbary_b(1) - \barbary_b(0) \xrightarrow[]{p} \frac{\mu^*_b(1)}{\omega_b} - \frac{\mu^*_b(0)}{\omega_b} = \beta_{1,b}^*.
\end{align*}

Hence in this setting our estimator is consistent.

\subsubsection{Asymptotic normality of one ratio}\label{subsec:asy_one_ratio}

We are interested in the asymptotic behavior of, for $t \in \{0, 1\}$,
\[\barbary_b(t)= \frac{1}{\sum_{j=1}^{m_b}w_{jb}T_{jb}}\sum_{j=1}^{m_b}T_{jb}w_{jb}\bar{Y}_{jb}(t).\]
We use the notation from Section~\ref{subsec:asy_one_ratio_main} but now add subscripts $b$ to indicate that we are referring to block $b$.
For $t \in \{1, 0\}$, let $C_{jb}(t) = \bar{Y}_{jb}(t)$, so that $\overline{\overline{c}}_b(t)=\barbary_b(t)$ and $\overline{\overline{C}}_b(t) = \barbarY_b(t)$.

\begin{corollary}\label{cor:one_ratio}
Under the conditions of Lemma~\ref{lem:main_res}, we have 
\begin{align*}
\frac{\overline{\overline{c}}_b(t) - \overline{\overline{C}}_b(t)}{\sqrt{\frac{1-f_b^t}{m_b^t}}S_{R,b}(t)/\wbar}
&= \frac{(\barbary_b(t)-\barbarY_b(t))}{\sqrt{\frac{1-f_b^t}{m_b^t}}\sqrt{\frac{1}{m_b-1}\sum_{j=1}^{m_b}(R_{jb}(t)/\wbar_b)^2}}\\
& \xrightarrow[]{d} N(0,1).
\end{align*}
\end{corollary}

The term in the denominator expands as follows:
\begin{align*}
\frac{1}{m_b-1}\sum_{j=1}^{m_b}(R_{jb}(t)/\wbar_b)^2 & = \frac{1}{m_b-1}\sum_{j=1}^{m_b}\left(\left(w_{jb}C_{jb}(t) - w_{jb}\overline{\overline{C}}_b(t)\right)/\wbar_b\right)^2\\
& = \frac{1}{(m_b-1)}\sum_{j=1}^{m_b}\frac{w_{jb}^2}{\wbar_b^2}\left(\bar{Y}_{jb}(t) -\barbarY_b(t)\right)^2.
\end{align*}

\subsubsection{Asymptotic normality}\label{subsec:asy_two_ratios}
We are interested in the ratio difference estimator, $\barbary_b(1) - \barbary_b(0)$.

Let $D_{jb}(t)=  (w_{jb}C_{jb}(t)-w_{jb}\overline{\overline{C}}_b(t))/\wbar_b = w_{jb}(\bar{Y}_{jb}(t)-\barbarY_b(t))/\wbar_b $, with $C_{jb}(t)$ defined as in Section~\ref{subsec:asy_one_ratio}.
Note that this means that we still have $\overline{\overline{c}}_b(t) =\barbary_b(t)$ and $\overline{\overline{C}}_b(t) = \barbarY_b(t)$.
Now we have reformulated our complicated difference in ratio estimators into a form that's easier to work with.

\begin{corollary}\label{cor:two_ratio}
Under the conditions of Lemma~\ref{lemma:two_ratio_main}, we have
\begin{align*}
\frac{\left(\overline{\overline{c}}_b(1) - \overline{\overline{C}}_b(1)\right) - \left(\overline{\overline{c}}_b(0)-\overline{\overline{C}}_b(0)\right)}{\sqrt{\text{Var}(\hat{D}_b)}} &= \frac{\left(\barbary_b(1) - \barbarY_b(1)\right) - \left(\barbary_b(0)-\barbarY_b(0)\right)}{\sqrt{\text{Var}(\hat{D}_b)}}\\
&  \xrightarrow[]{d} N(0, 1).
\end{align*}
\end{corollary}

The denominator simplifies as follows:
\begin{align*}
\text{Var}(\hat{D}_b) =&  \frac{S^2_{D,b}(1)}{m_b^1} + \frac{S^2_{D,b}(0)}{m_b^0} - \frac{S^2(D_b)}{m_b}\\
=&  \frac{\frac{1}{m_b-1}\sum_{j=1}^{m_b}D_{jb}^2(1)}{m_b^1} + \frac{\frac{1}{m_b-1}\sum_{j=1}^{m_b}D_{jb}^2(0)}{m_b^0} - \frac{\frac{1}{m_b-1}\sum_{j=1}^{m_b}(D_{jb}(1) - D_{jb}(0))^2}{m_b}\\
=&  \frac{\frac{1}{m_b-1}\sum_{j=1}^{m_b}\left(w_{jb}(\bar{Y}_{jb}(1)-\barbarY_b(1))/\wbar_b \right)^2}{m_b^1} + \frac{\frac{1}{m_b-1}\sum_{j=1}^{m_b}\left(w_{jb}(\bar{Y}_{jb}(0)-\barbarY_b(0))/\wbar_b \right)^2}{m_b^0}\\
& - \frac{\frac{1}{m_b-1}\sum_{j=1}^{m_b}\left(w_{jb}(\bar{Y}_{jb}(1)-\barbarY_b(1))/\wbar_b  - w_{jb}(\bar{Y}_{jb}(0)-\barbarY_b(0))/\wbar_b \right)^2}{m_b}\\
=&  \frac{\frac{1}{m_b-1}\sum_{j=1}^{m_b}\frac{w_{jb}^2}{\wbar_b^2}\left(\bar{Y}_{jb}(1)-\barbarY_b(1)\right)^2}{m_b^1} + \frac{\frac{1}{m_b-1}\sum_{j=1}^{m_b}\frac{w_{jb}^2}{\wbar_b^2}\left(\bar{Y}_{jb}(0)-\barbarY_b(0)) \right)^2}{m_b^0}\\
& - \frac{\frac{1}{m_b-1}\sum_{j=1}^{m_b}\frac{w_{jb}^2}{\wbar_b^2}\left(\bar{Y}_{jb}(1)-\bar{Y}_{jb}(0)  - (\barbarY_b(1)-\barbarY_b(0)) \right)^2}{m_b}.
\end{align*}

\subsection{Theorem 1 with covariates, multiple blocks}
In this section we explore the impact estimator from a weighted least squares regression with additional covariate adjustment.
We start by deriving the closed form of this estimator, which is $\hat{\beta}_{1,b} =
\barbary_b(1) - \barbary_b(0) -\left(\xbarbar^1_b - \xbarbar^0_b\right)\hat{\bm{\gamma}}$ for block $b$.
We then find consistency results and show the asymptotic normality result given in Theorem 1.
We do not specifically include interactions between the additional covariates and the block indicators as this immediately follows by generating the fully interacted set of covariates and then using our results on that extended set.

\subsubsection{The estimator}\label{subsec:cov_no_int_est}
We now have multiple blocks with covariate adjustment but not interactions between blocks and additional covariates in our regression.
However, we continue to include an interaction between block and treatment.
We assume that $m_b/m \xrightarrow[]{p} q_b$ where $0 < q_b < 1$.
Let there be $h$ blocks.
Accordingly, $\hat{\bm{\beta}}_1$ and $\hat{\bm{\beta}}_0$ are vectors with an entry for each block.
Thus we have $\hat{\beta}_{1,b}$, the $b$th entry of $\hat{\bm{\beta}}_1$, is still the treatment effect estimator for block $b$.

Now let $\tilde{\bm{x}}_{ijb} = \bm{x}_{ijb} - \xbarbar_b$ with 
\[\xbarbar_b= \frac{1}{w_b}\sum_{j=1}^{m_b}\sum_{i=1}^{n_{jb}}w_{ijb}\bm{x}_{ijb}\]
and $\tilde{\overline{\bm{x}}}_{jb} = \overline{\bm{x}}_{jb} - \xbarbar_b$ with
\[\overline{\bm{x}}_{jb}=\frac{1}{w_{jb}}\sum_{i=1}^{n_{jb}}w_{ijb}\bm{x}_{ijb}.\]

We also define
\[\xbarT_b = \frac{\sum_{j=1}^{m_b}w_{jb}T_{jb}\overline{\bm{x}}_{jb}}{\sum_{j=1}^{m_b}w_{jb}T_{jb}}\]
and
\[\xbarC_b = \frac{\sum_{j=1}^{m_b}w_{jb}(1-T_{jb})\overline{\bm{x}}_{jb}}{\sum_{j=1}^{m_b}w_{jb}(1-T_{jb})}.\]

We have $\bm{z}_{ijb} = (S_{ij1}\tilde{T}_{j1}, \dots, S_{ijh}\tilde{T}_{jh}, S_{ij1}, \dots,  S_{ijh}, \tilde{\bm{x}}_{ijb})$.
The estimated parameter vector is
\begin{align*}
\begin{pmatrix}
\hat{\bm{\beta}}_1\\
\hat{\bm{\beta}}_0
\end{pmatrix} = 
\left[\left(\sum_{b=1}^h\sum_{j=1}^{m_b}\sum_{i=1}^{n_{jb}}w_{ijb}\bm{z}'_{ijb}\bm{z}_{ijb}\right)^{-1}\sum_{b=1}^h\sum_{j=1}^{m_b}\sum_{i=1}^{n_{jb}}w_{ijb}\bm{z}'_{ijb}y_{ijb}\right].
\end{align*}

\textbf{Result:}
\[\hat{\beta}_{1,b} =
\barbary_b(1) - \barbary_b(0) -\left(\xbarbar^1_b - \xbarbar^0_b\right)\hat{\bm{\gamma}}\]

\begin{remark}
When we have interactions between blocks and additional covariates, we will have a different $\hat{\bm{\gamma}}$ for each block and the estimators for each block will be independent.
In that case, $\hat{\beta}_{1,b}$ is the same as the estimator as if we had only run the regression with block $b$, i.e. as if we only had one block.
Therefore, these results directly extend to that case.
\end{remark}

We have
\begin{align*}
&\begin{pmatrix}
\hat{\bm{\beta}}_1\\
\hat{\bm{\beta}}_0\\
\hat{\bm{\gamma}}\\
\end{pmatrix}\\ 
&=
\begin{pmatrix}
\sum_{j=1}^{m_1}w_{j1}\tilde{T}_{j1}^2 & \cdots& 0 & 0 &\cdots& 0 & \sum_{j=1}^{m_1}w_{j1}\tilde{T}_{j1}\tilde{\overline{\bm{x}}}_{j1}\\
\vdots  & \ddots & \vdots & \vdots & \ddots & \vdots & \vdots \\
0 &\cdots & \sum_{j=1}^{m_h}w_{jh}\tilde{T}_{jh}^2 & 0 & \cdots & 0 & \sum_{j=1}^{m_h}w_{jh}\tilde{T}_{jh}\tilde{\overline{\bm{x}}}_{jh}\\
0  & \cdots & 0 & w_1 & \cdots& 0 & \bm{0}_v\\
\vdots & \ddots & \vdots & \vdots & \ddots & \vdots & \bm{0}_v\\
0  & \cdots & 0 & 0 & \cdots & w_h & \bm{0}_v\\
\sum_{j=1}^{m_1}w_{j1}\tilde{T}_{j1}\tilde{\overline{\bm{x}}}_{j1}' & \cdots & \sum_{j=1}^{m_h}w_{jh}\tilde{T}_{jh}\tilde{\overline{\bm{x}}}_{jh}'  & \bm{0}_v' & \cdots & \bm{0}_v' & \sum_{b=1}^h\sum_{j=1}^{m_b}\sum_{i=1}^{n_{jb}}w_{ijb}\tilde{\bm{x}}_{ijb}'\tilde{\bm{x}}_{ijb}
\end{pmatrix}^{-1}\\
& \quad \times\begin{pmatrix}
\sum_{j=1}^{m_1}w_{j1}\tilde{T}_{j1}\bar{y}_{j1}\\
\vdots\\
\sum_{j=1}^{m_h}w_{jh}\tilde{T}_{jh}\bar{y}_{jh}\\
\sum_{j=1}^{m_1}w_{j1}\bar{y}_{j1}\\
\vdots\\
\sum_{j=1}^{m_h}w_{jh}\bar{y}_{jh}\\
\sum_{b=1}^h\sum_{j=1}^{m_b}\sum_{i=1}^{n_{jb}}w_{ijb}\tilde{\bm{x}}_{ijb}'y_{ijb}
\end{pmatrix}.
\end{align*}

We start by performing the matrix inversion.
We can break this matrix into the following blocks:
\begin{align*}
\bm{A} &= 
\begin{pmatrix}
\sum_{j=1}^{m_1}w_{j1}\tilde{T}_{j1}^2 & \cdots& 0 & 0 &\cdots& 0 \\
\vdots  & \ddots & \vdots & \vdots & \ddots & \vdots\\
0 &\cdots & \sum_{j=1}^{m_h}w_{jh}\tilde{T}_{jh}^2 & 0 & \cdots & 0\\
0  & \cdots & 0 & w_1 & \cdots& 0\\
\vdots & \ddots & \vdots & \vdots & \ddots & \vdots\\
0  & \cdots & 0 & 0 & \cdots & w_h 
\end{pmatrix}\\
&=\begin{pmatrix}
\frac{w^1_1w^0_1}{w_1} & \cdots& 0 & 0 &\cdots& 0 \\
\vdots  & \ddots & \vdots & \vdots & \ddots & \vdots\\
0 &\cdots & \frac{w^1_hw^0_h}{w_h} & 0 & \cdots & 0\\
0  & \cdots & 0 & w_1 & \cdots& 0\\
\vdots & \ddots & \vdots & \vdots & \ddots & \vdots\\
0  & \cdots & 0 & 0 & \cdots & w_h 
\end{pmatrix}
\end{align*}
\begin{align*}
\bm{B} &=
\begin{pmatrix}
\sum_{j=1}^{m_1}w_{j1}\tilde{T}_{j1}\tilde{\overline{\bm{x}}}_{j1}\\
\vdots\\
\sum_{j=1}^{m_h}w_{jh}\tilde{T}_{jh}\tilde{\overline{\bm{x}}}_{jh}\\
\bm{0}_v\\
\vdots\\
\bm{0}_v
\end{pmatrix}\\
&=
\begin{pmatrix}
\frac{w^1_1w^0_1}{w_1}\left(\xbarbar^1_1 - \xbarbar^0_1\right)\\
\vdots\\
\frac{w^1_hw^0_h}{w_h}\left(\xbarbar^1_h - \xbarbar^0_h\right)\\
\bm{0}_v\\
\vdots\\
\bm{0}_v
\end{pmatrix}\\
\bm{C} &=
\begin{pmatrix}
\sum_{j=1}^{m_1}w_{j1}\tilde{T}_{j1}\tilde{\overline{\bm{x}}}_{j1}' & \cdots & \sum_{j=1}^{m_h}w_{jh}\tilde{T}_{jh}\tilde{\overline{\bm{x}}}_{jh}' & \bm{0}_v' & \cdots & \bm{0}_v'
\end{pmatrix}\\
&=
\begin{pmatrix}
\frac{w^1_1w^0_1}{w_1}\left(\xbarbar^1_1 - \xbarbar^0_1\right)' & \cdots & \frac{w^1_hw^0_h}{w_h}\left(\xbarbar^1_h - \xbarbar^0_h\right)' & \bm{0}_v' & \cdots & \bm{0}_v'
\end{pmatrix}\\
\bm{D} &=  \left(\sum_{b=1}^h\sum_{j=1}^{m_b}\sum_{i=1}^{n_{jb}}w_{ijb}\tilde{\bm{x}}_{ijb}'\tilde{\bm{x}}_{ijb}\right).
\end{align*}

We can then use the following matrix inversion formula:
\begin{align}
\begin{pmatrix}
\bm{A} & \bm{B}\\
\bm{C} & \bm{D}
\end{pmatrix}^{-1} &=
\begin{pmatrix}
(\bm{A} - \bm{B}\bm{D}^{-1}\bm{C})^{-1} & -\bm{A}^{-1}\bm{B}(\bm{D} - \bm{C}\bm{A}^{-1}\bm{B})^{-1}\\
-\bm{D}^{-1}\bm{C}(\bm{A} - \bm{B}\bm{D}^{-1}\bm{C})^{-1}& (\bm{D} - \bm{C}\bm{A}^{-1}\bm{B})^{-1}
\end{pmatrix}.  \label{eq:mat_inv_2}
\end{align}

To make things easier, we will first derive a result for
\[\hat{\bm{\beta}}_1 +\begin{pmatrix}
\xbarbar^1_1 - \xbarbar^0_1\\
\vdots\\
\xbarbar^1_h - \xbarbar^0_h
\end{pmatrix} \hat{\bm{\gamma}}.\]

We have, based on our matrix inversion, that
\begin{align*}
\hat{\bm{\beta}}_1 =& \begin{pmatrix}
1 & \cdots & 0 & 0 & \cdots & 0 \\
\vdots & \ddots & \vdots & \vdots & \ddots& \vdots \\
0 & \cdots & 1 & 0 & \cdots& 0 
\end{pmatrix}
\begin{pmatrix}
(\bm{A} - \bm{B}\bm{D}^{-1}\bm{C})^{-1} & -\bm{A}^{-1}\bm{B}(\bm{D} - \bm{C}\bm{A}^{-1}\bm{B})^{-1}
\end{pmatrix}\\
&\times\begin{pmatrix}
\sum_{j=1}^{m_1}w_{j1}\tilde{T}_{j1}\bar{y}_{j1}\\
\vdots\\
\sum_{j=1}^{m_h}w_{jh}\tilde{T}_{jh}\bar{y}_{jh}\\
\sum_{j=1}^{m_1}w_{j1}\bar{y}_{j1}\\
\vdots\\
\sum_{j=1}^{m_h}w_{jh}\bar{y}_{jh}\\
\sum_{b=1}^h\sum_{j=1}^{m_b}\sum_{i=1}^{n_{jb}}w_{ijb}\tilde{\bm{x}}_{ijb}'y_{ijb}
\end{pmatrix}\\
=&\begin{pmatrix}
1 & \cdots & 0 & 0 & \cdots & 0 \\
\vdots & \ddots & \vdots & \vdots & \ddots& \vdots \\
0 & \cdots & 1 & 0 & \cdots& 0 
\end{pmatrix}
\begin{pmatrix}
(\bm{A} - \bm{B}\bm{D}^{-1}\bm{C})^{-1} & -\bm{A}^{-1}\bm{B}(\bm{D} - \bm{C}\bm{A}^{-1}\bm{B})^{-1}
\end{pmatrix}\\
&\times\begin{pmatrix}
\frac{w^1_1w^0_1}{w_1}\left(\barbary_1(1) - \barbary_1(0)\right)\\
\vdots\\
\frac{w^1_hw^0_h}{w_h}\left(\barbary_h(1) - \barbary_h(0)\right)\\
\sum_{j=1}^{m_1}w_{j1}\bar{y}_{j1}\\
\vdots\\
\sum_{j=1}^{m_h}w_{jh}\bar{y}_{jh}\\
\sum_{b=1}^h\sum_{j=1}^{m_b}\sum_{i=1}^{n_{jb}}w_{ijb}\tilde{\bm{x}}_{ijb}'y_{ijb}
\end{pmatrix}
\end{align*}
and
\begin{align*}
&\begin{pmatrix}
\xbarbar^1_1 - \xbarbar^0_1\\
\vdots\\
\xbarbar^1_h - \xbarbar^0_h
\end{pmatrix}\hat{\bm{\gamma}}\\
 & = 
\begin{pmatrix}
\xbarbar^1_1 - \xbarbar^0_1\\
\vdots\\
\xbarbar^1_h - \xbarbar^0_h
\end{pmatrix}
\begin{pmatrix}
-\bm{D}^{-1}\bm{C}(\bm{A} - \bm{B}\bm{D}^{-1}\bm{C})^{-1}& (\bm{D} - \bm{C}\bm{A}^{-1}\bm{B})^{-1}
\end{pmatrix}
\begin{pmatrix}
\frac{w^1_1w^0_1}{w_1}\left(\barbary_1(1) - \barbary_1(0)\right)\\
\vdots\\
\frac{w^1_hw^0_h}{w_h}\left(\barbary_h(1) - \barbary_h(0)\right)\\
\sum_{j=1}^{m_1}w_{j1}\bar{y}_{j1}\\
\vdots\\
\sum_{j=1}^{m_h}w_{jh}\bar{y}_{jh}\\
\sum_{b=1}^h\sum_{j=1}^{m_b}\sum_{i=1}^{n_{jb}}w_{ijb}\tilde{\bm{x}}_{ijb}'y_{ijb}
\end{pmatrix}.
\end{align*}

We see there are a lot of common terms when we add these expressions.
Let's first simplify
\[-\begin{pmatrix}
1 & \cdots & 0 & 0 & \cdots & 0 \\
\vdots & \ddots & \vdots & \vdots & \ddots& \vdots \\
0 & \cdots & 1 & 0 & \cdots& 0 
\end{pmatrix}\bm{A}^{-1}\bm{B}+ \begin{pmatrix}
\xbarbar^1_1 - \xbarbar^0_1\\
\vdots\\
\xbarbar^1_h - \xbarbar^0_h
\end{pmatrix}.\]

\begin{align*}
\begin{pmatrix}
1 & \cdots & 0 & 0 & \cdots & 0 \\
\vdots & \ddots & \vdots & \vdots & \ddots& \vdots \\
0 & \cdots & 1 & 0 & \cdots& 0 
\end{pmatrix}\bm{A}^{-1}\bm{B} =&
\begin{pmatrix}
1 & \cdots & 0 & 0 & \cdots & 0 \\
\vdots & \ddots & \vdots & \vdots & \ddots& \vdots \\
0 & \cdots & 1 & 0 & \cdots& 0 
\end{pmatrix}
\begin{pmatrix}
\frac{w_1}{w^1_1w^0_1} & \cdots & 0 & 0 & \cdots & 0\\
\vdots & \ddots & \vdots & \vdots & \ddots & \vdots\\
0 & \cdots & \frac{w_h}{w^1_hw^0_h} & 0 & \cdots & 0\\
0 & \cdots & 0 & \frac{1}{w_1} & \cdots & 0\\
\vdots & \ddots & \vdots & \vdots & \ddots & \vdots\\
0 & \cdots & 0 & 0 & \cdots & \frac{1}{w_h} 
\end{pmatrix}\\
&\times\begin{pmatrix}
\frac{w^1_1w^0_1}{w_1}\left(\xbarbar^1_1 - \xbarbar^0_1\right)\\
\vdots\\
\frac{w^1_hw^0_h}{w_h}\left(\xbarbar^1_h - \xbarbar^0_h\right)\\
\bm{0}_v\\
\vdots\\
\bm{0}_v
\end{pmatrix}\\
=&
\begin{pmatrix}
1 & \cdots & 0 & 0 & \cdots & 0 \\
\vdots & \ddots & \vdots & \vdots & \ddots& \vdots \\
0 & \cdots & 1 & 0 & \cdots& 0 
\end{pmatrix}
\begin{pmatrix}
\xbarbar^1_1 - \xbarbar^0_1\\
\vdots\\
\xbarbar^1_h - \xbarbar^0_h\\
\bm{0}_v\\
\vdots\\
\bm{0}_v
\end{pmatrix}\\
=&\begin{pmatrix}
\xbarbar^1_1 - \xbarbar^0_1\\
\vdots\\
\xbarbar^1_h - \xbarbar^0_h
\end{pmatrix}
\end{align*}
We see that 
\[-\begin{pmatrix}
1 & \cdots & 0 & 0 & \cdots & 0 \\
\vdots & \ddots & \vdots & \vdots & \ddots& \vdots \\
0 & \cdots & 1 & 0 & \cdots& 0 
\end{pmatrix}\bm{A}^{-1}\bm{B}+ \begin{pmatrix}
\xbarbar^1_1 - \xbarbar^0_1\\
\vdots\\
\xbarbar^1_h - \xbarbar^0_h
\end{pmatrix} = 
\begin{pmatrix}
\bm{0}_v\\
\vdots\\
\bm{0}_v
\end{pmatrix}.\]
Now we need to simplify
\[
-\begin{pmatrix}
1 & \cdots & 0 & 0 & \cdots & 0 \\
\vdots & \ddots & \vdots & \vdots & \ddots& \vdots \\
0 & \cdots & 1 & 0 & \cdots& 0 
\end{pmatrix} - \begin{pmatrix}
\xbarbar^1_1 - \xbarbar^0_1\\
\vdots\\
\xbarbar^1_h - \xbarbar^0_h
\end{pmatrix}
\bm{D}^{-1}\bm{C}.
\]

But first let's look at $\bm{A} - \bm{B}\bm{D}^{-1}\bm{C}$.
\begin{align*}
\bm{A} - \bm{B}\bm{D}^{-1}\bm{C}
 &=
\begin{pmatrix}
\frac{w^1_1w^0_1}{w_1} & \cdots & 0 & 0 & \cdots & 0\\
\vdots & \ddots & \vdots & \vdots & \ddots & \vdots\\
0 & \cdots & \frac{w^1_hw^0_h}{w_h} & 0 & \cdots & 0\\
0 & \cdots & 0 & w_1 & \cdots & 0\\
\vdots & \ddots & \vdots & \vdots & \ddots & \vdots\\
0 & \cdots & 0 & 0 & \cdots & w_h 
\end{pmatrix} - 
 \begin{pmatrix}
\frac{w^1_1w^0_1}{w_1}\left(\xbarbar^1_1 - \xbarbar^0_1\right)\\
\vdots\\
\frac{w^1_hw^0_h}{w_h}\left(\xbarbar^1_h - \xbarbar^0_h\right)\\
\bm{0}_v\\
\vdots\\
\bm{0}_v
\end{pmatrix} \bm{D}^{-1} \bm{C}\\
 &=
 \begin{pmatrix}
\frac{w^1_1w^0_1}{w_1} & \cdots & 0 & 0 & \cdots & 0\\
\vdots & \ddots & \vdots & \vdots & \ddots & \vdots\\
0 & \cdots & \frac{w^1_hw^0_h}{w_h} & 0 & \cdots & 0\\
0 & \cdots & 0 & w_1 & \cdots & 0\\
\vdots & \ddots & \vdots & \vdots & \ddots & \vdots\\
0 & \cdots & 0 & 0 & \cdots & w_h 
\end{pmatrix}
\left(\begin{pmatrix}
1 & \cdots & 0 & 0 & \cdots & 0\\
\vdots & \ddots & \vdots & \vdots & \ddots & \vdots\\
0 & \cdots & 1 & 0 & \cdots & 0\\
0 & \cdots & 0 & 1 & \cdots & 0\\
\vdots & \ddots & \vdots & \vdots & \ddots & \vdots\\
0 & \cdots & 0 & 0 & \cdots & 1 
\end{pmatrix} - 
 \begin{pmatrix}
\left(\xbarbar^1_1 - \xbarbar^0_1\right)\\
\vdots\\
\left(\xbarbar^1_h - \xbarbar^0_h\right)\\
\bm{0}_v\\
\vdots\\
\bm{0}_v
\end{pmatrix} \bm{D}^{-1} \bm{C}\right)
\end{align*}

Now let's return to \[
\begin{pmatrix}
1 & \cdots & 0 & 0 & \cdots & 0\\
\vdots & \ddots & \vdots & \vdots & \ddots & \vdots\\
0 & \cdots & 1 & 0 & \cdots & 0\\
0 & \cdots & 0 & 1 & \cdots & 0\\
\vdots & \ddots & \vdots & \vdots & \ddots & \vdots\\
0 & \cdots & 0 & 0 & \cdots & 1 
\end{pmatrix} - \begin{pmatrix}
\xbarbar^1_1 - \xbarbar^0_1\\
\vdots\\
\xbarbar^1_h - \xbarbar^0_h
\end{pmatrix}
\bm{D}^{-1}\bm{C}.
\]
\begin{align*}
\begin{pmatrix}
1 & \cdots & 0 & 0 & \cdots & 0\\
\vdots & \ddots & \vdots & \vdots & \ddots & \vdots\\
0 & \cdots & 1 & 0 & \cdots & 0\\
0 & \cdots & 0 & 1 & \cdots & 0\\
\vdots & \ddots & \vdots & \vdots & \ddots & \vdots\\
0 & \cdots & 0 & 0 & \cdots & 1 
\end{pmatrix} - \begin{pmatrix}
\xbarbar^1_1 - \xbarbar^0_1\\
\vdots\\
\xbarbar^1_h - \xbarbar^0_h
\end{pmatrix}
\bm{D}^{-1}\bm{C}&= 
\begin{pmatrix}
\frac{w_1}{w^1_1w^0_1} &\cdots& 0 & 0 & \cdots & 0 \\
\vdots &\ddots& \vdots & \vdots & \ddots & \vdots \\
0 & \cdots & \frac{w_2}{w^1_2w^0_2} & 0 & \cdots & 0 
\end{pmatrix} \left(\bm{A} - \bm{B}\bm{D}^{-1}\bm{C}\right)
\end{align*}
Hence,
\begin{align*}
\left(\begin{pmatrix}
1 & \cdots & 0 & 0 & \cdots & 0\\
\vdots & \ddots & \vdots & \vdots & \ddots & \vdots\\
0 & \cdots & 1 & 0 & \cdots & 0\\
0 & \cdots & 0 & 1 & \cdots & 0\\
\vdots & \ddots & \vdots & \vdots & \ddots & \vdots\\
0 & \cdots & 0 & 0 & \cdots & 1 
\end{pmatrix} - \begin{pmatrix}
\xbarbar^1_1 - \xbarbar^0_1\\
\vdots\\
\xbarbar^1_h - \xbarbar^0_h
\end{pmatrix}
\bm{D}^{-1}\bm{C}\right)\left(\bm{A} - \bm{B}\bm{D}^{-1}\bm{C}\right)^{-1} 
&= \begin{pmatrix}
\frac{w_1}{w^1_1w^0_1} &\cdots& 0 & 0 & \cdots & 0 \\
\vdots &\ddots& \vdots & \vdots & \ddots & \vdots \\
0 & \cdots & \frac{w_2}{w^1_2w^0_2} & 0 & \cdots & 0 
\end{pmatrix}.
\end{align*}

Putting it all together:
\begin{align*}
\hat{\bm{\beta}}_1 +\begin{pmatrix}
\xbarbar^1_1 - \xbarbar^0_1\\
\vdots\\
\xbarbar^1_h - \xbarbar^0_h
\end{pmatrix} \hat{\bm{\gamma}}
&=
\begin{pmatrix}
\frac{w_1}{w^1_1w^0_1} &\cdots& 0 & 0 & \cdots & 0 \\
\vdots &\ddots& \vdots & \vdots & \ddots & \vdots \\
0 & \cdots & \frac{w_2}{w^1_2w^0_2} & 0 & \cdots & 0 
\end{pmatrix}
\begin{pmatrix}
\frac{w^1_1w^0_1}{w_1}\left(\barbary_1(1) - \barbary_1(0)\right)\\
\vdots\\
\frac{w^1_hw^0_h}{w_h}\left(\barbary_h(1) - \barbary_h(0)\right)\\
\sum_{j=1}^{m_1}w_{j1}\bar{y}_{j1}\\
\vdots\\
\sum_{j=1}^{m_h}w_{jh}\bar{y}_{jh}\\
\sum_{b=1}^h\sum_{j=1}^{m_b}\sum_{i=1}^{n_{jb}}w_{ijb}\tilde{\bm{x}}_{ijb}'y_{ijb}
\end{pmatrix}\\
&=
\begin{pmatrix}
\barbary_1(1) - \barbary_1(0)\\
\vdots\\
\barbary_h(1) - \barbary_h(0)
\end{pmatrix}.
\end{align*}

Hence, we have the desired result,
\[\hat{\bm{\beta}}_1 =\begin{pmatrix}
\barbary_1(1) - \barbary_1(0)\\
\vdots\\
\barbary_h(1) - \barbary_h(0)
\end{pmatrix} -\begin{pmatrix}
\xbarbar^1_1 - \xbarbar^0_1\\
\vdots\\
\xbarbar^1_h - \xbarbar^0_h
\end{pmatrix} \hat{\bm{\gamma}}.\]

\subsubsection{Consistency}\label{subsec:cov_cons_no_int}

Here we will show that $\hat{\beta}_{1,b} \xrightarrow[]{p} \frac{1}{\omega_b}\left(\mu^*_{b}(1) - \mu^*_{b}(0)\right) = \beta_{1,b}^*$.
We assume the same limiting values for our average potential outcomes and weights as in Section~\ref{subsec:cons_no_cov}.
Assume that we have finite limiting values on the following weighted variance/covariance expressions, denoted as follows:
\[\bm{S}^2_{\bm{x}, b} = \frac{1}{m_b}\sum_{j=1}^{m_b}\sum_{i=1}^{n_{jb}}w_{ijb}\left(\bm{x}_{ijb}-\xbarbar_b\right)'\left(\bm{x}_{ijb}-\xbarbar_b\right) \xrightarrow[]{p} \bm{\Sigma}^2_{\bm{x}, b}\]
and
\[\bm{S}^2_{\bm{x}, Y, b}(t) = \frac{1}{m_b}\sum_{j=1}^{m_b}\sum_{i=1}^{n_{jb}}w_{ijb}\bm{x}_{ijb}'Y_{ijb}(t)-\xbarbar_b'\overline{wY(t)}_b \xrightarrow[]{p} \bm{\Sigma}^2_{\bm{x}, Y(t), b}.\]
Assume we also have a (unnamed) limiting value on the following variance expression:
\[\bm{S}^2_{\bm{x}Y, b}(t) = \frac{1}{m_b}\sum_{j=1}^{m_b}\sum_{i=1}^{n_{jb}}\left(w_{ijb}\bm{x}_{ijb}'Y_{ijb}(t)-\overline{wxY(t)}_b\right)^2\]
with
\[\overline{wxY(t)}_b =  \frac{1}{m_b}\sum_{j=1}^{m_b}\sum_{i=1}^{n_{jb}}w_{ijb}\bm{x}_{ijb}'Y_{ijb}(t).\]
Also assume that we have finite limiting values on the variances for the potential outcomes.
Further assume we have limiting values $\overline{\bm{X}}^*_{b}$, $\overline{\bm{X}'\bm{X}}^*_{b}$, and $\overline{\bm{X}\mu}_{wb}(t)$ such that $\frac{1}{m_b}\sum_{j=1}^{m_b}w_{jb}\overline{\bm{x}}_{jb} \xrightarrow[]{p} \overline{\bm{X}}^*_{b}$, $\frac{1}{m_b}\sum_{j=1}^{m_b}\sum_{i=1}^{n_{jb}}w_{ijb}\bm{x}_{ijb}'\bm{x}_{ijb} \xrightarrow[]{p} \overline{\bm{X}'\bm{X}}^*_{b}$, and $\frac{1}{m_b}\sum_{j=1}^{m_b}\sum_{i=1}^{n_{jb}}w_{ijb}\tilde{\bm{x}}_{ijb}'Y_{ijb}(t) \xrightarrow[]{p} \overline{\bm{X}\mu}^*_{b}(t)$.
Again, we assume that $m_b/m \xrightarrow[]{p} q_b$ where $0 < q_b < 1$, such that each block is growing to infinity with $m = \sum_{b=1}^hm_b$.
In Section~\ref{subsec:cons_no_cov}, we already showed that $\barbary_b(1) - \barbary_b(0) \xrightarrow[]{p} \mu^*_{b}(1) - \mu^*_{b}(0)$.
Thus, we now need to examine the extra term in this new $\hat{\beta}_{1,b}$, for which we need to simplify the asymptotic form of $\hat{\bm{\gamma}}$.

\begin{align*}
&\begin{pmatrix}
\hat{\bm{\beta}}_1\\
\hat{\bm{\beta}}_0\\
\hat{\bm{\gamma}}\\
\end{pmatrix}\\ 
=&
\begin{pmatrix}
\sum_{j=1}^{m_1}w_{j1}\tilde{T}_{j1}^2 & \cdots& 0 & 0 &\cdots& 0 & \sum_{j=1}^{m_1}w_{j1}\tilde{T}_{j1}\tilde{\overline{\bm{x}}}_{j1}\\
\vdots  & \ddots & \vdots & \vdots & \ddots & \vdots & \vdots \\
0 &\cdots & \sum_{j=1}^{m_h}w_{jh}\tilde{T}_{jh}^2 & 0 & \cdots & 0 & \sum_{j=1}^{m_h}w_{jh}\tilde{T}_{jh}\tilde{\overline{\bm{x}}}_{jh}\\
0  & \cdots & 0 & w_1 & \cdots& 0 & \bm{0}_v\\
\vdots & \ddots & \vdots & \vdots & \ddots & \vdots & \bm{0}_v\\
0  & \cdots & 0 & 0 & \cdots & w_h & \bm{0}_v\\
\sum_{j=1}^{m_1}w_{j1}\tilde{T}_{j1}\tilde{\overline{\bm{x}}}_{j1}' & \cdots & \sum_{j=1}^{m_h}w_{jh}\tilde{T}_{jh}\tilde{\overline{\bm{x}}}_{jh}'  & \bm{0}_v' & \cdots & \bm{0}_v' & \sum_{b=1}^h\sum_{j=1}^{m_b}\sum_{i=1}^{n_{jb}}w_{ijb}\tilde{\bm{x}}_{ijb}'\tilde{\bm{x}}_{ijb}
\end{pmatrix}^{-1}\\
&\times\begin{pmatrix}
\sum_{j=1}^{m_1}w_{j1}\tilde{T}_{j1}\bar{y}_{j1}\\
\vdots\\
\sum_{j=1}^{m_h}w_{jh}\tilde{T}_{jh}\bar{y}_{jh}\\
\sum_{j=1}^{m_1}w_{j1}\bar{y}_{j1}\\
\vdots\\
\sum_{j=1}^{m_h}w_{jh}\bar{y}_{jh}\\
\sum_{b=1}^h\sum_{j=1}^{m_b}\sum_{i=1}^{n_{jb}}w_{ijb}\tilde{\bm{x}}_{ijb}'y_{ijb}
\end{pmatrix}\\
=&
\begin{pmatrix}
\frac{1}{m_1}\sum_{j=1}^{m_1}w_{j1}\tilde{T}_{j1}^2 & \cdots& 0 & 0 &\cdots& 0 & \frac{1}{m_1}\sum_{j=1}^{m_1}w_{j1}\tilde{T}_{j1}\tilde{\overline{\bm{x}}}_{j1}\\
\vdots  & \ddots & \vdots & \vdots & \ddots & \vdots & \vdots \\
0 &\cdots & \frac{1}{m_h}\sum_{j=1}^{m_h}w_{jh}\tilde{T}_{jh}^2 & 0 & \cdots & 0 & \frac{1}{m_h}\sum_{j=1}^{m_h}w_{jh}\tilde{T}_{jh}\tilde{\overline{\bm{x}}}_{jh}\\
0  & \cdots & 0 & \frac{1}{m_1}w_1 & \cdots& 0 & \bm{0}_v\\
\vdots & \ddots & \vdots & \vdots & \ddots & \vdots & \bm{0}_v\\
0  & \cdots & 0 & 0 & \cdots & \frac{1}{m_h}w_h & \bm{0}_v\\
\frac{1}{m}\sum_{j=1}^{m_1}w_{j1}\tilde{T}_{j1}\tilde{\overline{\bm{x}}}_{j1}' & \cdots & \frac{1}{m}\sum_{j=1}^{m_h}w_{jh}\tilde{T}_{jh}\tilde{\overline{\bm{x}}}_{jh}'  & \bm{0}_v' & \cdots & \bm{0}_v' & \frac{1}{m}\sum_{b=1}^h\sum_{j=1}^{m_b}\sum_{i=1}^{n_{jb}}w_{ijb}\tilde{\bm{x}}_{ijb}'\tilde{\bm{x}}_{ijb}
\end{pmatrix}^{-1}\\
&\times\begin{pmatrix}
\frac{1}{m_1}\sum_{j=1}^{m_1}w_{j1}\tilde{T}_{j1}\bar{y}_{j1}\\
\vdots\\
\frac{1}{m_h}\sum_{j=1}^{m_h}w_{jh}\tilde{T}_{jh}\bar{y}_{jh}\\
\frac{1}{m_1}\sum_{j=1}^{m_1}w_{j1}\bar{y}_{j1}\\
\vdots\\
\frac{1}{m_h}\sum_{j=1}^{m_h}w_{jh}\bar{y}_{jh}\\
\frac{1}{m}\sum_{b=1}^h\sum_{j=1}^{m_b}\sum_{i=1}^{n_{jb}}w_{ijb}\tilde{\bm{x}}_{ijb}'y_{ijb}
\end{pmatrix}.
\end{align*}

We have 
\[\frac{1}{m_b}\sum_{j=1}^{m_b}w_{jb}\tilde{T}_{jb}^2 = \frac{1}{m_b}\frac{\wT_b\wC_b}{w_b} \xrightarrow[]{p} p_b(1-p_b)\omega_b,\]
\[ \frac{1}{m_b}\sum_{j=1}^{m_b}w_{jb} \xrightarrow[]{p} \omega_b,\]
and
\[\frac{1}{m_b}\sum_{j=1}^{m_b}w_{jb}\tilde{T}_{jb}\tilde{\overline{\bm{x}}}_{jb} = \frac{1}{m_b}\frac{\wC_b\wT_b}{w_b}\left(\xbarT_b-\xbarC_b\right) \xrightarrow[]{p} 0\]
because 
\begin{align*}
\xbarT_b-\xbarC_b \xrightarrow[]{p} \frac{\overline{\bm{X}}^*_{b}}{\omega_b} - \frac{\overline{\bm{X}}^*_{b}}{\omega_b} = 0.
\end{align*}
We also have
\begin{align*}
&\frac{1}{m}\sum_{b=1}^h\sum_{j=1}^{m_b}\sum_{i=1}^{n_{jb}}w_{ijb}\tilde{\bm{x}}_{ijb}'\tilde{\bm{x}}_{ijb}\\
&= \frac{1}{m}\sum_{b=1}^h\sum_{j=1}^{m_b}\sum_{i=1}^{n_{jb}}w_{ijb}\left(\bm{x}_{ijb}-\xbarbar_b\right)'\left(\bm{x}_{ijb}-\xbarbar_b\right)\\
&= \sum_{b=1}^h\frac{m_b}{m}\frac{1}{m_b}\sum_{j=1}^{m_b}\sum_{i=1}^{n_{jb}}w_{ijb}\left(\bm{x}_{ijb}-\xbarbar_b\right)'\left(\bm{x}_{ijb}-\xbarbar_b\right)\\
&\xrightarrow[]{p}  \sum_{b=1}^hq_b\bm{\Sigma}^2_{\bm{x}, b}.
\end{align*}

Then we have
\begin{align*}
&\begin{pmatrix}
\frac{1}{m_1}\sum_{j=1}^{m_1}w_{j1}\tilde{T}_{j1}^2 & \cdots& 0 & 0 &\cdots& 0 & \frac{1}{m_1}\sum_{j=1}^{m_1}w_{j1}\tilde{T}_{j1}\tilde{\overline{\bm{x}}}_{j1}\\
\vdots  & \ddots & \vdots & \vdots & \ddots & \vdots & \vdots \\
0 &\cdots & \frac{1}{m_h}\sum_{j=1}^{m_h}w_{jh}\tilde{T}_{jh}^2 & 0 & \cdots & 0 & \frac{1}{m_h}\sum_{j=1}^{m_h}w_{jh}\tilde{T}_{jh}\tilde{\overline{\bm{x}}}_{jh}\\
0  & \cdots & 0 & \frac{1}{m_1}w_1 & \cdots& 0 & \bm{0}_v\\
\vdots & \ddots & \vdots & \vdots & \ddots & \vdots & \bm{0}_v\\
0  & \cdots & 0 & 0 & \cdots & \frac{1}{m_h}w_h & \bm{0}_v\\
\frac{1}{m}\sum_{j=1}^{m_1}w_{j1}\tilde{T}_{j1}\tilde{\overline{\bm{x}}}_{j1}' & \cdots & \frac{1}{m}\sum_{j=1}^{m_h}w_{jh}\tilde{T}_{jh}\tilde{\overline{\bm{x}}}_{jh}'  & \bm{0}_v' & \cdots & \bm{0}_v' & \frac{1}{m}\sum_{b=1}^h\sum_{j=1}^{m_b}\sum_{i=1}^{n_{jb}}w_{ijb}\tilde{\bm{x}}_{ijb}'\tilde{\bm{x}}_{ijb}
\end{pmatrix}\\
&\xrightarrow[]{p}
\begin{pmatrix}
p_1(1-p_1)\omega_1 & \cdots & 0 & 0 & \cdots & 0 &  \bm{0}_v\\
\vdots & \ddots & \vdots & \vdots & \ddots & \vdots &  \vdots\\
0 & \cdots & p_h(1-p_h)\omega_h& 0 & \cdots & 0 &  \bm{0}_v\\
0 & \cdots & 0 & \omega_1 & \cdots & 0 & \bm{0}_v\\
\vdots & \ddots & \vdots & \vdots & \ddots & \vdots & \vdots\\
0 & \cdots & 0 & 0 & \cdots & \omega_h & \bm{0}_v\\
\bm{0}_v' & \cdots & \bm{0}_v' & \bm{0}_v' &  \cdots & \bm{0}_v' & \sum_{b=1}^hq_b\bm{\Sigma}^2_{\bm{x}, b}
\end{pmatrix}.
\end{align*}

Hence, by continuity of the inverse and Slutsky's theorem,
\begin{align*}
&\begin{pmatrix}
\frac{1}{m_1}\sum_{j=1}^{m_1}w_{j1}\tilde{T}_{j1}^2 & \cdots& 0 & 0 &\cdots& 0 & \frac{1}{m_1}\sum_{j=1}^{m_1}w_{j1}\tilde{T}_{j1}\tilde{\overline{\bm{x}}}_{j1}\\
\vdots  & \ddots & \vdots & \vdots & \ddots & \vdots & \vdots \\
0 &\cdots & \frac{1}{m_h}\sum_{j=1}^{m_h}w_{jh}\tilde{T}_{jh}^2 & 0 & \cdots & 0 & \frac{1}{m_h}\sum_{j=1}^{m_h}w_{jh}\tilde{T}_{jh}\tilde{\overline{\bm{x}}}_{jh}\\
0  & \cdots & 0 & \frac{1}{m_1}w_1 & \cdots& 0 & \bm{0}_v\\
\vdots & \ddots & \vdots & \vdots & \ddots & \vdots & \bm{0}_v\\
0  & \cdots & 0 & 0 & \cdots & \frac{1}{m_h}w_h & \bm{0}_v\\
\frac{1}{m}\sum_{j=1}^{m_1}w_{j1}\tilde{T}_{j1}\tilde{\overline{\bm{x}}}_{j1}' & \cdots & \frac{1}{m}\sum_{j=1}^{m_h}w_{jh}\tilde{T}_{jh}\tilde{\overline{\bm{x}}}_{jh}'  & \bm{0}_v' & \cdots & \bm{0}_v' & \frac{1}{m}\sum_{b=1}^h\sum_{j=1}^{m_b}\sum_{i=1}^{n_{jb}}w_{ijb}\tilde{\bm{x}}_{ijb}'\tilde{\bm{x}}_{ijb}
\end{pmatrix}^{-1}\\
&\xrightarrow[]{p}
\begin{pmatrix}
\frac{1}{p_1(1-p_1)\omega_1} & \cdots & 0 & 0 & \cdots & 0 &  \bm{0}_v\\
\vdots & \ddots & \vdots & \vdots & \ddots & \vdots &  \vdots\\
0 & \cdots & \frac{1}{p_h(1-p_h)\omega_h}& 0 & \cdots & 0 &  \bm{0}_v\\
0 & \cdots & 0 & \frac{1}{\omega_1} & \cdots & 0 & \bm{0}_v\\
\vdots & \ddots & \vdots & \vdots & \ddots & \vdots & \vdots\\
0 & \cdots & 0 & 0 & \cdots & \frac{1}{\omega_h }& \bm{0}_v\\
\bm{0}_v' & \cdots & \bm{0}_v' & \bm{0}_v' &  \cdots & \bm{0}_v' & \left(\sum_{b=1}^hq_b\bm{\Sigma}^2_{\bm{x}, b}\right)^{-1}
\end{pmatrix}.
\end{align*}

We also have 
\begin{align*}
\frac{1}{m_b}\sum_{j=1}^{m_b}w_{jb}\tilde{T}_{jb}\bar{y}_{jb} = \frac{1}{m_b}\frac{\wT_b\wC_b}{w_b}\left(\barbary_b(1)-\barbary_b(0)\right) &\xrightarrow[]{p} p_b(1-p_b)\omega_b\left(\frac{\mu^*_{b}(1)}{\omega_b}-\frac{\mu^*_{b}(0)}{\omega_b}\right)\\
& = p_b(1-p_b)\left(\mu^*_{b}(1)-\mu^*_{b}(0)\right)
\end{align*}
and
\[\frac{1}{m_b}\sum_{j=1}^{m_b}w_{jb}\bar{y}_{jb} = p_b\frac{1}{m_b^1}\sum_{j=1}^{m_b}w_{jb}T_{jb}\bar{Y}_{jb}(1) +  (1-p_b)\frac{1}{m_b^0}\sum_{j=1}^{m_b}w_{jb}(1-T_{jb})\bar{Y}_{jb}(0) \xrightarrow[]{p} p_b\mu^*_{b}(1) + (1-p_b)\mu^*_{b}(0).\]
Additionally,
\begin{align*}
&\frac{1}{m}\sum_{b=1}^h\sum_{j=1}^{m_b}\sum_{i=1}^{n_{jb}}w_{ijb}\tilde{\bm{x}}_{ijb}'y_{ijb}\\
 &=  \frac{1}{m}\sum_{b=1}^h\sum_{j=1}^{m_b}\sum_{i=1}^{n_{jb}}T_{jb}w_{ijb}(\bm{x}_{ijb}-\xbarbar_b)Y_{ijb}(1) + \frac{1}{m}\sum_{b=1}^h\sum_{j=1}^{m_b}\sum_{i=1}^{n_{jb}}(1-T_{jb})w_{ijb}(\bm{x}_{ijb}-\xbarbar_b)Y_{ijb}(0)\\
  &=  \sum_{b=1}^h\frac{m_b^1}{m}\left[\frac{1}{m_b^1}\sum_{j=1}^{m_b}\sum_{i=1}^{n_{jb}}T_{jb}w_{ijb}\bm{x}_{ijb}Y_{ijb}(1) - \xbarbar_b\left(\frac{1}{m_b^1}\sum_{j=1}^{m_b}T_{jb}w_{jb}\bar{Y}_{jb}(1)\right)\right]\\
&+\sum_{b=1}^h\frac{m_b^0}{m}\left[\frac{1}{m_b^0}\sum_{j=1}^{m_b}\sum_{i=1}^{n_{jb}}(1-T_{jb})w_{ijb}\bm{x}_{ijb}Y_{ijb}(0) - \xbarbar_b\left(\frac{1}{m_b^0}\sum_{j=1}^{m_b}(1-T_{jb})w_{jb}\bar{Y}_{jb}(0)\right)\right]\\
& \xrightarrow[]{p} \sum_{b=1}^hp_bq_b\bm{\Sigma}^2_{\bm{x}, Y(1), b} + \sum_{b=1}^h(1-p_b)q_b\bm{\Sigma}^2_{\bm{x}, Y(0), b}.
\end{align*}
This last line comes from the following two intermediate steps:
\begin{enumerate}[(a)]
\item Because we have limiting values on the variances of our potential outcomes, we have a law of large numbers type result (see Theorem B of \citet{scott1981asymptotic})
\[ \frac{1}{m_b^1}\sum_{j: T_{jb}  = t}\sum_{i=1}^{n_{jb}}w_{ijb}\bm{x}_{ijb}Y_{ijb}(t) - \xbarbar_b\left(\frac{1}{m_b^1}\sum_{j: T_{jb}  = t}w_{jb}\bar{Y}_{jb}(t)\right) -\bm{S}^2_{\bm{x}, Y, b}(t) \xrightarrow[]{p} 0.\]
\item This implies (see Lemma~\ref{lemma:simple_conv_res}) that
\[ \frac{1}{m_b^1}\sum_{j: T_{jb}  = t}\sum_{i=1}^{n_{jb}}w_{ijb}\bm{x}_{ijb}Y_{ijb}(t) - \xbarbar_b\left(\frac{1}{m_b^1}\sum_{j: T_{jb}  = t}w_{jb}\bar{Y}_{jb}(t)\right)  \xrightarrow[]{p} \bm{\Sigma}^2_{\bm{x}, Y(t), b}.\]
\end{enumerate}

Then we have
\begin{align*}
&\begin{pmatrix}
\frac{1}{m_1}\sum_{j=1}^{m_1}w_{j1}\tilde{T}_{j1}\bar{y}_{j1}\\
\vdots\\
\frac{1}{m_h}\sum_{j=1}^{m_h}w_{jh}\tilde{T}_{jh}\bar{y}_{jh}\\
\frac{1}{m_1}\sum_{j=1}^{m_1}w_{j1}\bar{y}_{j1}\\
\vdots\\
\frac{1}{m_h}\sum_{j=1}^{m_h}w_{jh}\bar{y}_{jh}\\
\frac{1}{m}\sum_{b=1}^h\sum_{j=1}^{m_b}\sum_{i=1}^{n_{jb}}w_{ijb}\tilde{\bm{x}}_{ijb}'y_{ijb}
\end{pmatrix}\\ &\xrightarrow[]{p}
\begin{pmatrix}
p_1(1-p_1)\left(\mu^*_{1}(1)-\mu^*_{1}(0)\right)\\
\vdots\\
p_h(1-p_h)\left(\mu^*_{h}(1)-\mu^*_{h}(0)\right)\\
p_1\mu^*_{1}(1) + (1-p_1)\mu^*_{1}(0)\\
\vdots\\
p_h\mu^*_{h}(1) + (1-p_h)\mu^*_{h}(0)\\
\sum_{b=1}^hp_bq_b\bm{\Sigma}^2_{\bm{x}, Y(1), b} + \sum_{b=1}^h(1-p_b)q_b\bm{\Sigma}^2_{\bm{x}, Y(0), b}
\end{pmatrix}.
\end{align*}

Putting this all together, we have the following result:
\begin{align*}
\begin{pmatrix}
\hat{\bm{\beta}}_1\\
\hat{\bm{\beta}}_0\\
\hat{\bm{\gamma}}\\
\end{pmatrix} \xrightarrow[]{p}&
\begin{pmatrix}
\frac{1}{p_1(1-p_1)\omega_1} & \cdots & 0 & 0 & \cdots & 0 &  \bm{0}_v\\
\vdots & \ddots & \vdots & \vdots & \ddots & \vdots &  \vdots\\
0 & \cdots & \frac{1}{p_h(1-p_h)\omega_h}& 0 & \cdots & 0 &  \bm{0}_v\\
0 & \cdots & 0 & \frac{1}{\omega_1} & \cdots & 0 & \bm{0}_v\\
\vdots & \ddots & \vdots & \vdots & \ddots & \vdots & \vdots\\
0 & \cdots & 0 & 0 & \cdots & \frac{1}{\omega_h }& \bm{0}_v\\
\bm{0}_v' & \cdots & \bm{0}_v' & \bm{0}_v' &  \cdots & \bm{0}_v' & \left(\sum_{b=1}^hq_b\bm{\Sigma}^2_{\bm{x}, b}\right)^{-1}
\end{pmatrix}\\
&\times\begin{pmatrix}
p_1(1-p_1)\left(\mu^*_{1}(1)-\mu^*_{1}(0)\right)\\
\vdots\\
p_h(1-p_h)\left(\mu^*_{h}(1)-\mu^*_{h}(0)\right)\\
p_1\mu^*_{1}(1) + (1-p_1)\mu^*_{1}(0)\\
\vdots\\
p_h\mu^*_{h}(1) + (1-p_h)\mu^*_{h}(0)\\
\sum_{b=1}^hp_bq_b\bm{\Sigma}^2_{\bm{x}, Y(1), b} + \sum_{b=1}^h(1-p_b)q_b\bm{\Sigma}^2_{\bm{x}, Y(0), b}
\end{pmatrix}\\
=&\begin{pmatrix}
\frac{1}{\omega_1}\left(\mu^*_{1}(1)-\mu^*_{1}(0)\right)\\
\vdots\\
\frac{1}{\omega_h}\left(\mu^*_{h}(1)-\mu^*_{h}(0)\right)\\
\frac{1}{\omega_1}\left(p_1\mu^*_{1}(1) + (1-p_1)\mu^*_{1}(0)\right)\\
\vdots\\
\frac{1}{\omega_h}\left(p_2\mu^*_{h}(1) + (1-p_h)\mu^*_{h}(0)\right)\\
\bm{\Gamma}
\end{pmatrix}\\
=&\begin{pmatrix}
\beta_{11}^*\\
\vdots\\
\beta_{1h}^*\\
\frac{1}{\omega_1}\left(p_1\mu^*_{1}(1) + (1-p_1)\mu^*_{1}(0)\right)\\
\vdots\\
\frac{1}{\omega_h}\left(p_h\mu^*_{h}(1) + (1-p_h)\mu^*_{h}(0)\right)\\
\bm{\Gamma}
\end{pmatrix},
\end{align*}
where 
\begin{align*}
\bm{\Gamma} =&\left(\sum_{b=1}^hq_b\bm{\Sigma}^2_{\bm{x}, b}\right)^{-1}\left(\sum_{b=1}^hp_bq_b\bm{\Sigma}^2_{\bm{x}, Y(1), b} + \sum_{b=1}^h(1-p_b)q_b\bm{\Sigma}^2_{\bm{x}, Y(0), b}\right).
\end{align*}

Hence, we have that $\hat{\beta}_{1,b} \xrightarrow[]{p} \beta_{1b}^*$.

\subsubsection{Asymptotic normality with known $\bm{\gamma}$}\label{subsec:covs_asy_norm_know_g_no_int}
We ultimately want to find asymptotic normality results for the estimator for a single block $b$ with covariate adjustment across blocks, $\hat{\beta}_{1,b} = \barbary_b(1) - \barbary_b(0) - \left(\xbarT_b-\xbarC_b\right)\hat{\bm{\gamma}}$.
But first, let us show, similar to \citet{li2017general}, that asymptotic normality holds for $\tilde{\beta}_{1b} = \barbary_b(1) - \barbary_b(0) - \left(\xbarT_b-\xbarC_b\right)\bm{\gamma}$, where $\bm{\gamma}$ is the finite population regression estimator we would have obtained if we had run the regression on the full schedule of potential outcomes, and is thus constant for each treatment arm.
Stated differently, we first find the result when we know $\bm{\gamma}$ and do not have to estimate it.
In particular, 
\[\bm{\gamma} = \left(\sum_{b=1}^hq_b\bm{S}^2_{\bm{x}, b}\right)^{-1}\left(\sum_{b=1}^hp_bq_b\bm{S}^2_{\bm{x}, Y, b}(1) + \sum_{b=1}^h(1-p_b)q_b\bm{S}^2_{\bm{x}, Y, b}(0)\right).\]

In this setting, we have that $C_{jb}(t) = (\bar{Y}_{jb}(t)-\overline{\bm{x}}_{jb}\bm{\gamma})$ and so $D_{jb}(t) = w_{jb}(\bar{Y}_{jb}(t)-\overline{\bm{x}}_{jb}\bm{\gamma}-\barbarY_b(t)+\xbarbar_b\bm{\gamma})/\wbar_b$.
Hence $\overline{\overline{c}}_b(t)= \barbary_b(t) - \xbarT_b\bm{\gamma}$ and $\overline{\overline{C}}_b(t) = \barbarY_b(t) - \xbarbar_b\bm{\gamma}$.

\begin{corollary}\label{cor:known_gamma_no_int}
Under the conditions of Lemma~\ref{lemma:two_ratio_main} for block $b$,
\begin{align*}
\frac{\left(\overline{\overline{c}}_b(1) - \overline{\overline{c}}_b(0) \right) - \left(\overline{\overline{C}}_b(1) - \overline{\overline{C}}_b(0)\right)}{\sqrt{\text{Var}(\hat{D}_b)}} &= \frac{\tilde{\beta}_{1b} - \left(\barbarY_b(1) - \barbarY_b(0)\right)}{\sqrt{\text{Var}(\hat{D}_b)}}\\ &\xrightarrow[]{d} N(0, 1).
\end{align*}
\end{corollary}

\begin{proof}
The result is a direct consequence of Lemma~\ref{lemma:two_ratio_main}.
The equality holds by noting that
\begin{align*}
\frac{\sum_{j=1}^{m_b}w_{jb}\left(\bar{Y}_{jb}(1)-\overline{\bm{x}}_{jb}\bm{\gamma}\right)}{\sum_{j=1}^{m_b}w_{jb}} - \frac{\sum_{j=1}^{m_b}w_{jb}\left(\bar{Y}_{jb}(0)-\overline{\bm{x}}_{jb}\bm{\gamma}\right)}{\sum_{j=1}^{m_b}w_{jb}} = \barbarY_b(1) - \barbarY_b(0).
\end{align*}
\end{proof}

The variance in the denominator of our asymptotic expression simplifies as follows:
\begin{align*}
\text{Var}(\hat{D_b}) =&  \frac{S^2_{D,b}(1)}{m^1} + \frac{S^2_{D,b}(0)}{m^0} - \frac{S^2(D_b)}{m}\\
=&  \frac{\frac{1}{m_b-1}\sum_{j=1}^{m_b}D_{jb}^2(1)}{m_b^1} + \frac{\frac{1}{m_b-1}\sum_{j=1}^{m_b}D_{jb}^2(0)}{m_b^0} - \frac{\frac{1}{m_b-1}\sum_{j=1}^{m_b}(D_{jb}(1) - D_{jb}(0))^2}{m_b}\\
=&  \frac{\frac{1}{m_b-1}\sum_{j=1}^{m_b}\frac{w_{jb}^2}{\wbar_b^2}\left(\bar{Y}_{jb}(1)-\barbarY_b(1)-\left(\overline{\bm{x}}_{jb}-\xbarbar_b\right)\bm{\gamma}\right)^2}{m_b^1}\\
& + \frac{\frac{1}{m_b-1}\sum_{j=1}^{m_b}\frac{w_{jb}^2}{\wbar_b^2}\left(\bar{Y}_{jb}(0)-\barbarY_b(0)-\left(\overline{\bm{x}}_{jb} - \xbarbar_b\right)\bm{\gamma}\right)^2}{m_b^0}\\
& - \frac{\frac{1}{m_b-1}\sum_{j=1}^{m_b}\frac{w_{jb}^2}{\wbar_b^2}\left(\bar{Y}_{jb}(1)-\bar{Y}_{jb}(0) - (\barbarY_b(1)-\barbarY_b(0))\right)^2}{m_b}\\
=&  \frac{\frac{1}{m_b-1}\sum_{j=1}^{m_b}\frac{w_{jb}^2}{\wbar_b^2}\left(\bar{Y}_{jb}(1)-\barbarY_b(1)-\tilde{\overline{\bm{x}}}_{jb} \bm{\gamma}\right)^2}{m_b^1} + \frac{\frac{1}{m_b-1}\sum_{j=1}^{m_b}\frac{w_{jb}^2}{\wbar_b^2}\left(\bar{Y}_{jb}(0)-\barbarY_b(0)-\tilde{\overline{\bm{x}}}_{jb} \bm{\gamma}\right)^2}{m_b^0}\\
& - \frac{\frac{1}{m_b-1}\sum_{j=1}^{m_b}\frac{w_{jb}^2}{\wbar_b^2}\left(\bar{Y}_{jb}(1)-\bar{Y}_{jb}(0) - (\barbarY_b(1)-\barbarY_b(0))\right)^2}{m_b}.
\end{align*}

\subsubsection{Joint asymptotic normality with known $\bm{\gamma}$}\label{subsec:covs_asy_norm_know_g_no_int_joint}

We now examine the joint convergence of $\tilde{\bm{\beta}}_{1}$.
We see that each element of $\tilde{\bm{\beta}}_{1}$, $\tilde{\beta}_{1b}$, is independent.
Thus, we have the characteristic function for $A_b = \frac{\tilde{\beta}_{1b} - \left(\barbarY_b(1) - \barbarY_b(0)\right)}{\sqrt{\text{Var}(\hat{D}_b)}}$, which we denote $\phi_b(t)$, converges as follows:
\[\phi_b(t) \xrightarrow[]{} e^{-t^2/2}\]
by the equivalency of convergence in distribution and point-wise convergence of characteristic functions and using Corollary~\ref{cor:known_gamma_no_int}. 

Now take any linear combination $L$, defined as
\[L = \sum_{b=1}^h s_b A_b.\]
Then the characteristic function of $L$ is
\begin{align*}
\phi_L(t) &= E\left[e^{itL}\right]\\
&= E\left[e^{it\sum_{b=1}^h s_b A_b}\right]\\
&= E\left[\prod_{b=1}^he^{it s_b A_b}\right]\\
&= \prod_{b=1}^hE\left[e^{it s_b A_b}\right]\\
&= \prod_{b=1}^h\phi_b(s_bt)\\
&\xrightarrow[]{} \prod_{b=1}^he^{-s_b^2t^2/2}\\
&=e^{-(\sum_{b=1}^hs_b^2)t^2/2}.
\end{align*}

Hence, $L \xrightarrow[]{d} N(0, \sum_{b=1}^hs_b^2)$.
Therefore, by the Cramer-Wold device, 
\begin{align*}
\begin{pmatrix}
 \frac{\tilde{\beta}_{11} - \left(\barbarY_1(1) - \barbarY_1(0)\right)}{\sqrt{\text{Var}(\hat{D}_1)}}\\
 \vdots\\
  \frac{\tilde{\beta}_{1h} - \left(\barbarY_h(1) - \barbarY_h(0)\right)}{\sqrt{\text{Var}(\hat{D}_h)}}
\end{pmatrix}
\xrightarrow[]{d} N(\bm{0}_h, \bm{I}_h).
\end{align*}

\subsubsection{Theorem 1: Asymptotic normality with estimated $\hat{\bm{\gamma}}$}\label{asy_norm_cov_no_int}

We now move to our primary result, Theorem 1.
To obtain this result, we first find asymptotic normality results for $\hat{\beta}_{1,b} = \barbary_b(1) - \barbary_b(0) - \left(\xbarT_b-\xbarC_b\right)\hat{\bm{\gamma}}$ for a single block $b$.
To do this, we need to show $\hat{\beta}_{1,b}$ has the same asymptotic distribution as $\tilde{\beta}_{1,b}= \barbary_b(1) - \barbary_b(0) - \left(\xbarT_b-\xbarC_b\right)\bm{\gamma}$, as done in \citet{li2017general} for the unweighted case.
Following that paper, we aim to show that the difference is order $o_p(m_b^{-1/2})$.

First note that we can use results from Section~\ref{subsec:asy_one_ratio} to put a convergence rate on $\xbarT_b-\xbarC_b$.
Convergence in probability of each element of $\xbarT_b$ and $\xbarC_b$ implies convergence of the entire vector.
Hence we can look at one entry of $\xbarT_b$ and $\xbarC_b$ at a time (since we can use component-wise convergence in probability).
For the $k$th component, let $C_{jb, k}(t) = \left[\overline{\bm{x}}_{jb}\right]_k$, noting that this does not change under treatment or control.
Then under the conditions of Lemma~\ref{lem:main_res} we have asymptotic normality results for each of the components of $\xbarT_b$ and $\xbarC_b$.
Assuming a limiting value on the variance, this in turn means that we have, for the $k$th component, $\left[\xbarT_b-\xbarbar_b\right]_k = O_p(m_b^{-1/2})$ and $\left[\xbarC_b-\xbarbar_b\right]_k  = O_p(m_b^{-1/2})$.

Then we have 
\begin{align*}
\left[\xbarT_b-\xbarC_b\right]_k&= \left[\xbarT_b - \xbarbar_b-\left(\xbarC_b - \xbarbar_b\right)\right]_k\\
&=O_p(m_b^{-1/2}).
\end{align*}

Now note that
\begin{align*}
\hat{\beta}_{1,b} &= \barbary_b(1) - \barbary_b(0) - \left(\xbarT_b-\xbarC_b\right)\hat{\bm{\gamma}}\\
&= \frac{\sum_{j=1}^{m_b}T_{jb}w_{jb}\left(\bar{y}_{jb}-\overline{\bm{x}}_{jb}\bm{\gamma}\right)}{\sum_{j=1}^{m_b}T_{jb}w_{jb}} - \frac{\sum_{j=1}^{m_b}(1-T_{jb})w_{jb}\left(\bar{y}_{jb}-\overline{\bm{x}}_{jb}\bm{\gamma}\right)}{\sum_{j=1}^{m_b}(1-T_{jb})w_{jb}} - \left(\xbarT_b-\xbarC_b\right)\left(\hat{\bm{\gamma}}-\bm{\gamma}\right).
\end{align*}

From the limiting value assumptions in Section~\ref{subsec:cov_cons_no_int}, $\hat{\bm{\gamma}}-\bm{\gamma} \xrightarrow[]{p} 0$ and so $\left(\xbarT_b-\xbarC_b\right)\left(\hat{\bm{\gamma}}-\bm{\gamma}\right) = o_p(m_b^{-1/2})$.
This means that $\hat{\beta}_{1,b}$  has the same asymptotic distribution as $\tilde{\beta}_{1b}$ and so we can use Corollary~\ref{cor:known_gamma_no_int}.

In Theorem 1, we assume the conditions of Corollary~\ref{cor:known_gamma_no_int}, the conditions of Lemma~\ref{lem:main_res} applied to each of the components of $\xbarT_b$ and $\xbarC_b$ as well as limiting values on the variance expression of each component.
Also assume limiting values on the following variance expressions:
\begin{align*}
\bm{S}^2_{\bm{x}, b} &= \frac{1}{m_b}\sum_{j=1}^{m_b}\sum_{i=1}^{n_{jb}}w_{ijb}\left(\bm{x}_{ijb}-\xbarbar_b\right)'\left(\bm{x}_{ijb}-\xbarbar_b\right)\\
\bm{S}^2_{\bm{x}, Y, b}(t) &= \frac{1}{m_b}\sum_{j=1}^{m_b}\sum_{i=1}^{n_{jb}}w_{ijb}\bm{x}_{ijb}'Y_{ijb}(t)-\xbarbar_b'\overline{wY(t)}_b\\
\bm{S}^2_{\bm{x}Y, b}(t) &= \frac{1}{m_b}\sum_{j=1}^{m_b}\sum_{i=1}^{n_{jb}}\left(w_{ijb}\bm{x}_{ijb}'Y_{ijb}(t)-\overline{wxY(t)}_b\right)^2
\text{for $t \in \{0, 1\}$ with}&\\
\overline{wxY(t)}_b & =  \frac{1}{m_b}\sum_{j=1}^{m_b}\sum_{i=1}^{n_{jb}}w_{ijb}\bm{x}_{ijb}'Y_{ijb}(t).
\end{align*}
Then we have the result of Theorem 1,
\begin{align*}
\frac{\hat{\beta}_{1,b} - \left(\barbarY_b(1) - \barbarY_b(0)\right)}{\sqrt{\text{Var}(\hat{D}_b)}} &\xrightarrow[]{d} N(0, 1).
\end{align*}
with 
\begin{align*}
\text{Var}(\hat{D_b}) 
=&  \frac{\frac{1}{m_b-1}\sum_{j=1}^{m_b}\frac{w_{jb}^2}{\wbar_b^2}\left(\bar{Y}_{jb}(1)-\barbarY_b(1)-\tilde{\overline{\bm{x}}}_{jb} \bm{\gamma}\right)^2}{m_b^1} + \frac{\frac{1}{m_b-1}\sum_{j=1}^{m_b}\frac{w_{jb}^2}{\wbar_b^2}\left(\bar{Y}_{jb}(0)-\barbarY_b(0)-\tilde{\overline{\bm{x}}}_{jb} \bm{\gamma}\right)^2}{m_b^0}\\
& - \frac{\frac{1}{m_b-1}\sum_{j=1}^{m_b}\frac{w_{jb}^2}{\wbar_b^2}\left(\bar{Y}_{jb}(1)-\bar{Y}_{jb}(0) - (\barbarY_b(1)-\barbarY_b(0))\right)^2}{m_b}.
\end{align*}

\subsubsection{Joint asymptotic normality with estimated $\hat{\bm{\gamma}}$}\label{asy_norm_cov_no_int_joint}

We now investigate joint asymptotic normality results for the vector $\hat{\bm{\beta}}_{1}$.
Each element of $\hat{\bm{\beta}}_{1}$, $\hat{\beta}_{1,b}$, is dependent because of the shared $\hat{\gamma}$ term.

With block-covariate interactions, each $\hat{\beta}_{1,b}$ is independent and thus the joint asymptotic result is immediate.

Without these interactions, we first define $C_b = \frac{\hat{\beta}_{1,b} - \left(\barbarY_b(1) - \barbarY_b(0)\right)}{\sqrt{\text{Var}(\hat{D}_b)}}$.
Based on results from Section~\ref{asy_norm_cov_no_int}, we have that $C_b$ converges to the same distribution that 
\[A_b = \frac{\tilde{\beta}_{1b} - \left(\barbarY_b(1) - \barbarY_b(0)\right)}{\sqrt{\text{Var}(\hat{D}_b)}}\]
converges to.
This implies that the linear combination
\[S = \sum_{b=1}^h s_b C_b\]
has the same asymptotic distribution as
\[L = \sum_{b=1}^h s_b A_b.\]
In particular,
\[S = \sum_{b=1}^h s_b\frac{\hat{\beta}_{1,b} - \left(\barbarY_b(1) - \barbarY_b(0)\right)}{\sqrt{\text{Var}(\hat{D}_b)}} = \sum_{b=1}^h s_bA_b - \sum_{b=1}^h s_b\left(\xbarT_b-\xbarC_b\right)\left(\hat{\bm{\gamma}}-\bm{\gamma}\right)\]
and the last sum is a finite sum of terms that are $o_p(m_b^{-1/2})$.



Hence, $S \xrightarrow[]{d} N(0, \sum_{b=1}^hs_b^2)$.
Therefore, by the Cramer-Wold device, 
\begin{align*}
\begin{pmatrix}
 \frac{\hat{\beta}_{1,1} - \left(\barbarY_1(1) - \barbarY_1(0)\right)}{\sqrt{\text{Var}(\hat{D}_1)}}\\
 \vdots\\
  \frac{\hat{\beta}_{1,h} - \left(\barbarY_h(1) - \barbarY_h(0)\right)}{\sqrt{\text{Var}(\hat{D}_h)}}
\end{pmatrix}
\xrightarrow[]{d} N(\bm{0}_h, \bm{I}_h).
\end{align*}

\subsection{Restricted model}
In this section, we explore the impact estimator from a weighted least squares regression with additional covariate adjustment and without interactions between block indicators and treatment indicators.
We now have a single treatment effect estimator that aggregates across blocks.
We first find the closed form for this estimator as 
$$\hat{\beta}_1= \frac{\sum_{b=1}^h\frac{w_b^1w_b^0}{w_b}\left(\barbary_b(1)-\barbary_b(0)\right)}{\sum_{b=1}^h\frac{w^1_bw^0_b}{w_b}} -  \frac{\sum_{b=1}^h\frac{w_b^1w_b^0}{w_b}\left(\xbarbar_b^1-\xbarbar_b^0\right)}{\sum_{b=1}^h\frac{w^1_bw^0_b}{w_b}}\hat{\bm{\gamma}} . $$
We then find consistency results and show an asymptotic normality result.
As before, we first assume a known $\bm{\gamma}$ to focus on the pooling of individual treatment impacts across blocks and then extend to an estimated $\hat{\bm{\gamma}}$.

\subsubsection{The estimator}
We now examine the model with no interactions between treatment and blocks .
In this case, there is a single treatment effect estimator, $\hat{\beta}_1$ for all blocks.
We have $\bm{z}_{ijb} = (\tilde{T}_{jb}, S_{ij1}, \dots, S_{ijh}, \tilde{\bm{x}}_{ijb})$.
The estimated parameter vector is
\begin{align*}
\begin{pmatrix}
\hat{\beta}_1\\
\hat{\bm{\beta}}_0\\
\hat{\bm{\gamma}}
\end{pmatrix} = 
\left[\left(\sum_{b=1}^{h}\sum_{j=1}^{m_b}\sum_{i=1}^{n_{jb}}w_{ijb}\bm{z}'_{ijb}\bm{z}_{ijb}\right)^{-1}\sum_{b=1}^{h}\sum_{j=1}^{m_b}\sum_{i=1}^{n_{jb}}w_{ijb}\bm{z}'_{ijb}y_{ijb}\right].
\end{align*}

Using the same techniques as in Section~\ref{subsec:cov_no_int_est} we find the following:

\noindent \textbf{Result:}
\[\hat{\beta}_1= \frac{\sum_{b=1}^h\frac{w_b^1w_b^0}{w_b}\left(\barbary_b(1)-\barbary_b(0)\right)}{\sum_{b=1}^h\frac{w^1_bw^0_b}{w_b}} -  \frac{\sum_{b=1}^h\frac{w_b^1w_b^0}{w_b}\left(\xbarbar_b^1-\xbarbar_b^0\right)}{\sum_{b=1}^h\frac{w^1_bw^0_b}{w_b}}\hat{\bm{\gamma}}.\]

Derivations available upon request.

\subsubsection{Consistency}
\textbf{Result:} 
\begin{align*}
\hat{\beta}_1 &= \frac{\sum_{b=1}^h\frac{w_b^1w_b^0}{w_b}\left(\barbary_b(1)-\barbary_b(0)\right)}{\sum_{b=1}^h\frac{w^1_bw^0_b}{w_b}} -  \frac{\sum_{b=1}^h\frac{w_b^1w_b^0}{w_b}\left(\xbarbar_b^1-\xbarbar_b^0\right)}{\sum_{b=1}^h\frac{w^1_bw^0_b}{w_b}} \hat{\bm{\gamma}}\\
& \xrightarrow[]{p}\frac{\sum_{b=1}^hq_bp_b(1-p_b)\left(\mu^*_b(1)-\mu^*_b(0)\right)}{\sum_{b=1}^hq_bp_b(1-p_b)\omega_b}
\end{align*}

Using results and limiting values given in Section~\ref{subsec:cons_no_cov} and Section~\ref{subsec:cov_cons_no_int} and also $m_b/m \xrightarrow[]{p} q_b$ ($0<q_b<1$), we have 
\begin{align*}
\begin{pmatrix}
\hat{\beta}_1\\
\hat{\bm{\beta}}_0\\
\hat{\bm{\gamma}}
\end{pmatrix}
=&
\begin{pmatrix}
\frac{1}{m}\sum_{b=1}^h\frac{w^1_bw^0_b}{w_b} & 0 & \cdots & 0 & \frac{1}{m}\sum_{b=1}^h\frac{w_b^1w_b^0}{w_b}\left(\xbarbar_b^1-\xbarbar_b^0\right) \\
0 &\frac{1}{m_1}w_1 & \cdots & 0 & \bm{0}_v\\
\vdots & \vdots & \ddots & \vdots & \vdots\\
0 & 0 & \cdots & \frac{1}{m_h}w_h & \bm{0}_v\\
\frac{1}{m}\sum_{b=1}^h\frac{w_b^1w_b^0}{w_b}\left(\xbarbar_b^{1'}-\xbarbar_b^{0'}\right) & \bm{0}_v' & \cdots & \bm{0}_v' &\frac{1}{m}\sum_{b=1}^h\sum_{j=1}^{m_b}\sum_{i=1}^{n_{jb}}w_{ijb}\tilde{\bm{x}}_{ijb}'\tilde{\bm{x}}_{ijb}
\end{pmatrix}^{-1}\\
&\times\begin{pmatrix}
\frac{1}{m}\sum_{b=1}^h\frac{w_b^1w_b^0}{w_b}\left(\barbary_b(1)-\barbary_b(0)\right)\\
\frac{1}{m_1}\left(w_1^1\barbary_1(1)+w_1^0\barbary_1(0)\right)\\
\vdots\\
\frac{1}{m_h}\left(w_h^1\barbary_h(1)+w_h^0\barbary_h(0)\right)\\
\frac{1}{m}\sum_{b=1}^h\sum_{j=1}^{m_b}\sum_{i=1}^{n_{jb}}w_{ijb}\tilde{\bm{x}}_{ijb}'\bar{y}_{ijb}
\end{pmatrix}\\
\xrightarrow[]{p}&
\begin{pmatrix}
\sum_{b=1}^hq_bp_b(1-p_b)\omega_b & 0 & \cdots & 0 & 0 \\
0 &\omega_1 & \cdots & 0 & \bm{0}_v\\
\vdots & \vdots & \ddots & \vdots & \vdots\\
0 & 0 & \cdots & \omega_h & \bm{0}_v\\
0 & \bm{0}_v' & \cdots & \bm{0}_v' &\sum_{b=1}^hq_b\bm{\Sigma}^2_{\bm{x}, b}
\end{pmatrix}^{-1}\\
&\times\begin{pmatrix}
\sum_{b=1}^hq_bp_b(1-p_b)\left(\mu^*_b(1)-\mu^*_b(0)\right)\\
p_1\mu_1(1)+(1-p_1)\mu_1(0)\\
\vdots\\
p_h\mu_h(1)+(1-p_h)\mu_h(0)\\
\sum_{b=1}^hp_bq_b\bm{\Sigma}^2_{\bm{x}, Y(1), b} + \sum_{b=1}^h(1-p_b)q_b\bm{\Sigma}^2_{\bm{x}, Y(0), b}
\end{pmatrix}\\
=& \begin{pmatrix}
\frac{\sum_{b=1}^hq_bp_b(1-p_b)\left(\mu^*_b(1)-\mu^*_b(0)\right)}{\sum_{b=1}^hq_bp_b(1-p_b)\omega_b}\\
\frac{1}{\omega_1}\left(p_1\mu_1(1)+(1-p_1)\mu_1(0)\right)\\
\vdots\\
\frac{1}{\omega_h}\left(p_h\mu_h(1)+(1-p_h)\mu_h(0)\right)\\
\bm{\Gamma}
\end{pmatrix}
\end{align*}
where
\begin{align*}
\bm{\Gamma} =& \left(\sum_{b=1}^hq_b\bm{\Sigma}^2_{\bm{x}, b}\right)^{-1}
\left(\sum_{b=1}^hp_bq_b\bm{\Sigma}^2_{\bm{x}, Y(1), b} + \sum_{b=1}^h(1-p_b)q_b\bm{\Sigma}^2_{\bm{x}, Y(0), b}\right).
\end{align*}

\subsubsection{Asymptotic normality with known $\bm{\gamma}$}
Let $\tilde{q}_b = m_b/m$.
Following Section~\ref{subsec:covs_asy_norm_know_g_no_int_joint}, we want to start by finding the asymptotic distribution of 
\begin{align*}
\tilde{\beta}_1 &= \frac{\sum_{b=1}^h\frac{w_b^1w_b^0}{w_b}\left(\barbary_b(1)-\barbary_b(0)\right)}{\sum_{a=1}^h\frac{w^1_aw^0_a}{w_a}} -  \frac{\sum_{b=1}^h\frac{w_b^1w_b^0}{w_b}\left(\xbarbar_b^1-\xbarbar_b^0\right)}{\sum_{a=1}^h\frac{w^1_aw^0_a}{w_a}} \bm{\gamma}\\
&= \sum_{b=1}^h\frac{\tilde{q}_bp_b\wTbar_b(\wbar_b- p_b\wTbar_b)/\wbar_b}{ \sum_{a=1}^h\tilde{q}_ap_a\wTbar_a(\wbar_a - p_a\wTbar_a)/\wbar_a}\left(\barbary_b(1)-\barbary_b(0) - \left(\xbarbar_b^1-\xbarbar_b^0\right)\bm{\gamma}\right).
\end{align*}
Further denote
\begin{align*}
\beta_1 &= \sum_{b=1}^h\frac{\tilde{q}_bp_b(1-p_b)\wbar_b}{ \sum_{a=1}^h\tilde{q}_ap_a(1-p_a)\wbar_a}\left(\barbarY_b(1)-\barbarY_b(0)\right).
\end{align*}

Let 
\[U_{jb}(t) = w_{jb}\left(\bar{Y}_{jb}(t) - \overline{\bm{x}}_{jb}\bm{\gamma}\right),\quad \bar{U}_b(t) = \frac{1}{m_b}\sum_{j=1}^{m_b}U_{jb}(t),\quad \text{and} \quad \bar{u}_b(t) = \frac{1}{m_b^t}\sum_{j: T_{jb}=t}U_{jb}(t).\]
Then we can write
\begin{align*}
\tilde{\beta}_1 &= \sum_{b=1}^h\frac{\tilde{q}_bp_b\wTbar_b(\wbar_b- p_b\wTbar_b)/\wbar_b}{ \sum_{a=1}^h\tilde{q}_ap_a\wTbar_a(\wbar_a - p_a\wTbar_a)/\wbar_a}\left(\frac{\bar{u}_b(1)}{\wTbar}-\frac{\bar{u}_b(0)}{\wCbar}\right).
\end{align*}
It is useful to rewrite this in terms of the fewest possible random variables, so we rewrite $\wTbar_b = (\wbar_b-p_b\wTbar_b)/(1-p_b)$ everywhere as follows:
\begin{align*}
\tilde{\beta}_1 &= \sum_{b=1}^h\frac{\tilde{q}_bp_b\wTbar_b(\wbar_b- p_b\wTbar_b)/\wbar_b}{ \sum_{a=1}^h\tilde{q}_ap_a\wTbar_a(\wbar_a - p_a\wTbar_a)/\wbar_a}\left(\frac{\bar{u}_b(1)}{\wTbar}-\frac{(1-p_b)\bar{u}_b(0)}{\wbar_b-p_b\wTbar_b}\right).
\end{align*}

Let, for $t \in \{0, 1\}$ and $z \in \{w, U\}$,
\[a_{z,b}(t) = \text{max}_{1 \leq j \leq m_b}\left(z_{j,b}(t)-\bar{z}_{b}(t)\right)^2\]
and
\[v_{z,b}(t) = \frac{1}{m_b-1}\sum_{j=1}^{m_b}\left(z_{j,b}(t)-\bar{z}_{b}(t)\right)^2,\]
noting that $w_{jb}(t) = w_{jb}$.
Let $\bm{t} = (\bar{u}_1(1), \bar{u}_1(0), \wTbar_1, \cdots, \bar{u}_h(1), \bar{u}_h(0), \wTbar_h)$ and \\
$\bm{T} = (\bar{U}_1(1), \bar{U}_1(0), \wbar_1, \cdots, \bar{U}_h(1), \bar{U}_h(0), \wbar_h)$.

\begin{theorem}\label{theor:comb_true_gam}
Let us assume the following conditions:
\begin{enumerate}[(a)]
\item As $m \to \infty$, 
\[\text{max}_{1\leq b \leq h}\text{max}_{z \in \{w, U\}}\text{max}_{t \in \{0, 1\}}\frac{a_{z,b}(t)}{p_b(1-p_b)m_bv_{z,b}(t)}  \to 0.\]
\item The correlation matrix of $\bm{t}$ has a limiting value $\bm{\Sigma}$.
\item We have limiting values on the following variance expressions: $m\text{Var}(\wTbar_b)$ and $m\text{Var}(\bar{u}_b(z))$ for all $b \in \{1,\dots,h\}$ and $z \in \{0, 1\}$.
\item $\bar{U}_b(1) \neq 0$ or $ \bar{U}_b(0) \neq 0$ for some $b$.
\end{enumerate}
Then we have
\begin{align*}
\frac{\tilde{\beta}_1 - \beta_1}{\sqrt{\text{Var}\left(\tilde{\beta}_1\right)}} \xrightarrow[]{d} \text{N}(0,1)
\end{align*}
where
\begin{align*}
\text{Var}\left(\tilde{\beta}_1\right)  = \sum_{b=1}^h&\frac{1}{m_b(m_b-1)}\sum_{j=1}^{m_b}\Bigg( \frac{\tilde{q}_bp_b(1-p_b)(1-2p_b)}{\sqrt{p_b(1-p_b)}\left(\sum_{a=1}^h\tilde{q}_ap_a(1-p_a)\wbar_a\right)}\left(\beta_{1,b} - \beta_1\right)(w_{jb} - \wbar_b)\\
& + \sqrt{\frac{1-p_b}{p_b}}\frac{\tilde{q}_bp_b(1-p_b)\wbar_b}{\sum_{a=1}^h\tilde{q}_ap_a(1-p_a)\wbar_a}\left(\frac{U_{jb}(1)}{\wbar_b} -\frac{w_{jb}}{\wbar_b} \frac{\bar{U}_b(1)}{\wbar_b}\right)\\
& - \sqrt{\frac{p_b}{1-p_b}}\frac{\tilde{q}_bp_b(1-p_b)\wbar_b}{\sum_{a=1}^h\tilde{q}_ap_a(1-p_a)\wbar_a}\left(\frac{U_{jb}(0)}{\wbar_b} - \frac{w_{jb}}{\wbar_b}\frac{\bar{U}_b(0)}{\wbar_b}\right) \Bigg)^2\\
 = \sum_{b=1}^h&\frac{1}{m_b(m_b-1)}\sum_{j=1}^{m_b}\Bigg( \frac{\tilde{q}_bp_b(1-p_b)(1-2p_b)}{\sqrt{p_b(1-p_b)}\left(\sum_{a=1}^h\tilde{q}_ap_a(1-p_a)\wbar_a\right)}\left(\beta_{1,b} - \beta_1\right)(w_{jb} - \wbar_b)\\
& + \sqrt{\frac{1-p_b}{p_b}}\frac{\tilde{q}_bp_b(1-p_b)\wbar_b}{\sum_{a=1}^h\tilde{q}_ap_a(1-p_a)\wbar_a}\left(\frac{w_{jb}\left(\bar{Y}_{jb}(1) - \overline{\bm{x}}_{jb}\bm{\gamma}\right)}{\wbar_b} -\frac{w_{jb}\left(\barbarY_b(1) - \xbarbar_b\bm{\gamma}\right)}{\wbar_b} \right)\\
& - \sqrt{\frac{p_b}{1-p_b}}\frac{\tilde{q}_bp_b(1-p_b)\wbar_b}{\sum_{a=1}^h\tilde{q}_ap_a(1-p_a)\wbar_a}\left(\frac{w_{jb}\left(\bar{Y}_{jb}(0) - \overline{\bm{x}}_{jb}\bm{\gamma}\right)}{\wbar_b} - \frac{w_{jb}\left(\barbarY_b(0) - \xbarbar_b\bm{\gamma}\right)}{\wbar_b}\right) \Bigg)^2.
\end{align*}
\end{theorem}

\begin{proof}
By Theorem 4 of \citet{li2017general}, if as $m \to \infty$
\[\text{max}_{1\leq b \leq h}\text{max}_{z \in \{w, U\}}\text{max}_{t \in \{0, 1\}}\frac{a_{z,b}(t)}{p_b(1-p_b)m_bv_{z,b}(t)}  \to 0\]
and the correlation matrix of $\bm{t}$ has a limiting value $\bm{\Sigma}$, then
\[\left(\frac{\bar{u}_1(1)-\bar{U}_1(1)}{\sqrt{\text{Var}(\bar{u}_1(1))}}, \frac{\bar{u}_1(0)-\bar{U}_1(0)}{\sqrt{\text{Var}(\bar{u}_1(0))}}, \frac{\wTbar_1-\wbar_1}{\sqrt{\text{Var}(\wTbar_1)}}, \cdots,  \frac{\bar{u}_h(0)-\bar{U}_h(0)}{\sqrt{\text{Var}(\bar{u}_h(0))}}, \frac{\wTbar_h-\wbar_h}{\sqrt{\text{Var}(\wTbar_h)}}\right) \xrightarrow[]{d} N(\bm{0},\bm{\Sigma}). \]
To use the delta method given in \citet{pashley2019delta}, we also require that $\bm{t}-\bm{T} \xrightarrow[]{p} 0$ (i.e., $\wTbar_b-\wbar_b \xrightarrow[]{p} 0$ and $\bar{u}_b(z)-\bar{U}_b(z) \xrightarrow[]{p} 0$ for all $b \in \{1,\dots,h\}$ and $z \in \{0, 1\}$).
This is satisfied by our assumption on limiting values of the variances (this can be seen directly from our prior results or Markov's inequality).

We can determine the variance and covariances by noting that blocks are independent of each other, but random variables within blocks are dependent.
We have 
\[\text{Var}\left(\bar{u}_b(t)\right) = \left(\frac{1}{m_b^t}-\frac{1}{m_b}\right)v_{z,U}(t),\]
\[\text{Var}\left(\wtbar_b \right) = \left(\frac{1}{m_b^t}-\frac{1}{m_b}\right)v_{z,w}(t),\]
\[\text{Cov}\left(\bar{u}_b(1), \bar{u}_b(0)\right) = -\frac{1}{m_b}\frac{1}{m_b-1}\sum_{j=1}^{m_b}\left(U_{jb}(1)-\bar{U}_b(1)\right)\left(U_{jb}(0)-\bar{U}_b(0)\right),\]
\[\text{Cov}\left(\bar{u}_b(1), \wTbar_b\right) =  \left(\frac{1}{m_b^1}-\frac{1}{m_b}\right)\frac{1}{m_b-1}\sum_{j=1}^{m_b}\left(U_{jb}(1)-\bar{U}_b(1)\right)\left(w_{jb}-\wbar_b\right),\]
and
\[\text{Cov}\left(\bar{u}_b(0), \wTbar_b\right) = -\frac{1}{m_b}\frac{1}{m_b-1}\sum_{j=1}^{m_b}\left(U_{jb}(0)-\bar{U}_b(0)\right)\left(w_{jb}-\wbar_b\right).\]

Now $g(\cdot): \mathbb{R}^{3h} \to \mathbb{R}$ takes our vector $\bm{t}$ and returns the estimator $\tilde{\beta}_1$.
As our weights are all positive and non-zero, this function is continuous and differential on the domain of $\bm{T}$.
We have
\begin{align*}
\nabla g(\bm{T}) &=
\begin{pmatrix}
\frac{\partial g(\bm{T})}{\partial \bar{u}_1(1)} &
\frac{\partial g(\bm{T})}{\partial \bar{u}_1(0)}&
\frac{\partial g(\bm{T})}{\partial \wTbar_1}&
\hdots&
\frac{\partial g(\bm{T})}{\partial \bar{u}_h(1)}&
\frac{\partial g(\bm{T})}{\partial \bar{u}_h(0)}&
\frac{\partial g(\bm{T})}{\partial \wTbar_h}
\end{pmatrix}^T.
\end{align*}
The partial derivatives are
\begin{align*}
\frac{\partial \hat{\beta}}{\partial \bar{u}_1(1)}\Big|_{\bm{T}} =& \frac{\tilde{q}_bp_b(1-p_b)\wbar_b}{\sum_{a=1}^h\tilde{q}_ap_a(1-p_a)\wbar_a}\frac{1}{\wbar_b}\\
\frac{\partial \hat{\beta}}{\partial\bar{u}_1(0)}\Big|_{\bm{T}} =& -\frac{\tilde{q}_bp_b(1-p_b)\wbar_b}{\sum_{a=1}^h\tilde{q}_ap_a(1-p_a)\wbar_a}\frac{1}{\wbar_b}\\
\frac{\partial \hat{\beta}}{\partial \wbar_b}\Big|_{\bm{T}} =& -\tilde{q}_bp_b(1-2p_b)\sum_{c=1}^h\frac{\tilde{q}_cp_c(1-p_c)\wbar_c}{\left(\sum_{a=1}^h\tilde{q}_ap_a(1-p_a)\wbar_a\right)^2}\left(\frac{\bar{U}_c(1)}{\wbar_c} - \frac{\bar{U}_c(0)}{\wbar_c}\right)\\
 &- \frac{\tilde{q}_bp_b(1-p_b)}{\sum_{a=1}^h\tilde{q}_ap_a(1-p_a)\wbar_a}\left(\frac{p_b}{1-p_b}\frac{\bar{U}_b(1)}{\wbar_b} + \frac{\bar{U}_b(0)}{\wbar_b} \right)\\
 =& -\frac{\tilde{q}_bp_b(1-2p_b)}{\sum_{a=1}^h\tilde{q}_ap_a(1-p_a)\wbar_a}\beta_1- \frac{\tilde{q}_bp_b(1-p_b)\wbar_b}{\sum_{a=1}^h\tilde{q}_ap_a(1-p_a)\wbar_a}\left(\frac{p_b}{1-p_b}\frac{\bar{U}_b(1)}{(\wbar_b)^2} + \frac{\bar{U}_b(0)}{(\wbar_b)^2} \right).
\end{align*}

We see that all of these derivatives are continuous and, assuming that we do not have $\bar{U}_b(0) = \bar{U}_b(0) =0$ for all $b$, we can therefore use the delta method result of Theorem 2 of \citet{pashley2019delta}, which gives
\begin{align*}
\frac{g(\bm{t}) - g(\bm{T})}{\sqrt{\left(\nabla g(\bm{T})\right)^T \bm{V}\bm{\Sigma}\bm{V}\nabla g(\bm{T})}} \xrightarrow[]{d} \text{N}(0,1)
\end{align*}
where
\begin{align*}
\bm{V} &= \begin{pmatrix}
\sqrt{\text{Var}(\bar{u}_1(1))} & 0 & \cdots & 0\\
0 & \sqrt{\text{Var}(\bar{u}_1(0))}  & \cdots & 0\\
\vdots & \vdots & \ddots & \vdots\\
0 & 0 & \cdots & \sqrt{\text{Var}(\wTbar_h)}
\end{pmatrix}.
\end{align*}

The denominator corresponds to the variance of $\tilde{\beta}_1$.
Let $\beta_{1,b} = (\bar{U}_b(1)-\bar{U}_b(0))/\wbar_b = \barbarY_b(1)-\barbarY_b(0)$.
We can now do the multiplication to get the variance, noting that symmetry can greatly simplify the calculations.
Replacing the limiting values of $\bm{\Sigma}$ with the sample values, we get
\begin{align*}
\text{Var}(\tilde{\beta}_1)
 = \sum_{b=1}^h&\frac{1}{m_b(m_b-1)}\sum_{j=1}^{m_b}\Bigg( -\frac{\sqrt{1-p_b}\tilde{q}_bp_b(1-2p_b)}{\sqrt{p_b}\left(\sum_{a=1}^h\tilde{q}_ap_a(1-p_a)\wbar_a\right)}\beta_1(w_{jb} - \wbar_b)\\
 &- \sqrt{\frac{1-p_b}{p_b}}\frac{\tilde{q}_bp_b(1-p_b)\wbar_b}{\sum_{a=1}^h\tilde{q}_ap_a(1-p_a)\wbar_a}\left(\frac{p_b}{1-p_b}\frac{\bar{U}_b(1)}{(\wbar_b)^2} + \frac{\bar{U}_b(0))}{(\wbar_b)^2} \right)(w_{jb} - \wbar_b)\\
& + \sqrt{\frac{1-p_b}{p_b}}\frac{\tilde{q}_bp_b(1-p_b)\wbar_b}{\sum_{a=1}^h\tilde{q}_ap_a(1-p_a)\wbar_a}\left(\frac{U_{jb}(1)}{\wbar_b} - \frac{\bar{U}_b(1)}{\wbar_b}\right)\\
& + \sqrt{\frac{p_b}{1-p_b}}\frac{\tilde{q}_bp_b(1-p_b)\wbar_b}{\sum_{a=1}^h\tilde{q}_ap_a(1-p_a)\wbar_a}\left(\frac{U_{jb}(0)}{\wbar_b} - \frac{\bar{U}_b(0)}{\wbar_b}\right) \Bigg)^2.
\end{align*}
We can rewrite this more digestibly in terms of block estimators as
\begin{align*}
\text{Var}(\tilde{\beta}_1)
 = \sum_{b=1}^h&\frac{1}{m_b(m_b-1)}\sum_{j=1}^{m_b}\Bigg( \frac{\tilde{q}_bp_b(1-p_b)(1-2p_b)}{\sqrt{p_b(1-p_b)}\left(\sum_{a=1}^h\tilde{q}_ap_a(1-p_a)\wbar_a\right)}\left(\beta_{1,b} - \beta_1\right)(w_{jb} - \wbar_b)\\
& + \sqrt{\frac{1-p_b}{p_b}}\frac{\tilde{q}_bp_b(1-p_b)\wbar_b}{\sum_{a=1}^h\tilde{q}_ap_a(1-p_a)\wbar_a}\left(\frac{w_{jb}\left(\bar{Y}_{jb}(1) - \overline{\bm{x}}_{jb}\bm{\gamma}\right)}{\wbar_b} -\frac{w_{jb}\left(\barbarY_b(1) - \xbarbar_b\bm{\gamma}\right)}{\wbar_b} \right)\\
& + \sqrt{\frac{p_b}{1-p_b}}\frac{\tilde{q}_bp_b(1-p_b)\wbar_b}{\sum_{a=1}^h\tilde{q}_ap_a(1-p_a)\wbar_a}\left(\frac{w_{jb}\left(\bar{Y}_{jb}(0) - \overline{\bm{x}}_{jb}\bm{\gamma}\right)}{\wbar_b} - \frac{w_{jb}\left(\barbarY_b(0) - \xbarbar_b\bm{\gamma}\right)}{\wbar_b}\right) \Bigg)^2.
\end{align*}
\end{proof}

\subsubsection{Theorem 2: Asymptotic normality with estimated $\hat{\bm{\gamma}}$}\label{asy_norm_comb_est}
Finally, to prove Theorem 2 we can follow the same proof as in Section~\ref{asy_norm_cov_no_int}.

Theorem 2 assume the conditions of Theorem~\ref{theor:comb_true_gam}, the conditions of Lemma~\ref{lem:main_res} applied to each of the components of $\xbarT_b$ and $\xbarC_b$, and a limiting value on the asymptotic variance for those components.
It also assumes limiting values on the following variance expressions:
\begin{align*}
\bm{S}^2_{\bm{x}, b} &= \frac{1}{m_b}\sum_{j=1}^{m_b}\sum_{i=1}^{n_{jb}}w_{ijb}\left(\bm{x}_{ijb}-\xbarbar_b\right)'\left(\bm{x}_{ijb}-\xbarbar_b\right)\\
\bm{S}^2_{\bm{x}, Y, b}(t) &= \frac{1}{m_b}\sum_{j=1}^{m_b}\sum_{i=1}^{n_{jb}}w_{ijb}\bm{x}_{ijb}'Y_{ijb}(t)-\xbarbar_b'\overline{wY(t)}_b\\
\bm{S}^2_{\bm{x}Y, b}(t) &= \frac{1}{m_b}\sum_{j=1}^{m_b}\sum_{i=1}^{n_{jb}}\left(w_{ijb}\bm{x}_{ijb}'Y_{ijb}(t)-\overline{wxY(t)}_b\right)^2\
\text{for $t \in \{0, 1\}$ with}&\\
\overline{wxY(t)}_b & =  \frac{1}{m_b}\sum_{j=1}^{m_b}\sum_{i=1}^{n_{jb}}w_{ijb}\bm{x}_{ijb}'Y_{ijb}(t).
\end{align*}
Then, we have the result of Theorem 2,
\begin{align*}
\frac{\hat{\beta}_{1} - \beta_1}{\sqrt{\text{Var}\left(\tilde{\beta}_1\right)}} &\xrightarrow[]{d} N(0, 1).
\end{align*}
with 
\begin{align*}
\text{Var}\left(\hat{\beta}_1\right)  = \sum_{b=1}^h&\frac{1}{m_b(m_b-1)}\sum_{j=1}^{m_b}\Bigg( \frac{\tilde{q}_bp_b(1-p_b)(1-2p_b)}{\sqrt{p_b(1-p_b)}\left(\sum_{a=1}^h\tilde{q}_ap_a(1-p_a)\wbar_a\right)}\left(\beta_{1,b} - \beta_1\right)(w_{jb} - \wbar_b)\\
& + \sqrt{\frac{1-p_b}{p_b}}\frac{\tilde{q}_bp_b(1-p_b)\wbar_b}{\sum_{a=1}^h\tilde{q}_ap_a(1-p_a)\wbar_a}\left(\frac{U_{jb}(1)}{\wbar_b} -\frac{w_{jb}}{\wbar_b} \frac{\bar{U}_b(1)}{\wbar_b}\right)\\
& + \sqrt{\frac{p_b}{1-p_b}}\frac{\tilde{q}_bp_b(1-p_b)\wbar_b}{\sum_{a=1}^h\tilde{q}_ap_a(1-p_a)\wbar_a}\left(\frac{U_{jb}(0)}{\wbar_b} - \frac{w_{jb}}{\wbar_b}\frac{\bar{U}_b(0)}{\wbar_b}\right) \Bigg)^2\\
 = \sum_{b=1}^h&\frac{1}{m_b(m_b-1)}\sum_{j=1}^{m_b}\Bigg( \frac{\tilde{q}_bp_b(1-p_b)(1-2p_b)}{\sqrt{p_b(1-p_b)}\left(\sum_{a=1}^h\tilde{q}_ap_a(1-p_a)\wbar_a\right)}\left(\beta_{1,b} - \beta_1\right)(w_{jb} - \wbar_b)\\
& + \sqrt{\frac{1-p_b}{p_b}}\frac{\tilde{q}_bp_b(1-p_b)\wbar_b}{\sum_{a=1}^h\tilde{q}_ap_a(1-p_a)\wbar_a}\left(\frac{w_{jb}\left(\bar{Y}_{jb}(1) - \overline{\bm{x}}_{jb}\bm{\gamma}\right)}{\wbar_b} -\frac{w_{jb}\left(\barbarY_b(1) - \xbarbar_b\bm{\gamma}\right)}{\wbar_b} \right)\\
& + \sqrt{\frac{p_b}{1-p_b}}\frac{\tilde{q}_bp_b(1-p_b)\wbar_b}{\sum_{a=1}^h\tilde{q}_ap_a(1-p_a)\wbar_a}\left(\frac{w_{jb}\left(\bar{Y}_{jb}(0) - \overline{\bm{x}}_{jb}\bm{\gamma}\right)}{\wbar_b} - \frac{w_{jb}\left(\barbarY_b(0) - \xbarbar_b\bm{\gamma}\right)}{\wbar_b}\right) \Bigg)^2.
\end{align*}

\subsection{Asymptotic Efficiency Losses Using the Individual Data}
Corollary 1 to Theorem 1 in the main text shows that aggregating
(averaging) the data to the cluster level produces asymptotically
efficient ATE estimators. To quantify efficiency losses using the
individual data, we first express the \(\bm{\gamma}\) parameter as a
weighted average of between- and within-cluster population regression
parameters, \(\bm{\gamma}_{\bm{B}}\) and
\(\bm{\gamma}_{\bm{W}}\), using
\(\bm{\gamma} = \bm{\mathit{\Gamma}}_{\bm{x}}\bm{\gamma}_{B}+(\bm{I}_{{v\times v}}-\bm{\mathit{\Gamma}}_{\bm{x}})\bm{\gamma}_{\bm{W}}\).
In this expression, \(\bm{\gamma}_{\bm{B}}\) is defined
analogously to \(\bm{\gamma}\) in Section A.3.3 in the supplementary
materials (and in Remark 2 of the main text) by replacing
\(\bm{S}_{\bm{x},b}^{2}\) with
$\bm{S}_{\bm{x},b,\bm{B}}^{2} = \frac{1}{m_{b}}\sum_{j = 1}^{m_{b}}w_{jb}\left( {\overline{\bm{x}}}_{jb} - \xbarbar_{b} \right)^{\prime}\left( {\overline{\bm{x}}}_{jb} -\xbarbar_{b} \right)$,
and using parallel cluster-level versions of
\(\bm{S}_{\bm{x}, Y, b}^{2}\left( 1 \right)\)
and
\(\bm{S}_{\bm{x},Y,b}^{2}\left( 0 \right)\).
This is the regression parameter that would be obtained using the full
schedule of potential outcomes and covariates aggregated to the cluster
level. We define \(\bm{\gamma}_{\bm{W}}\) similarly, for
example, by replacing
\(\bm{S}_{\bm{x},b}^{2}\) with
\(\bm{S}_{\bm{x},b,\bm{W}}^{2} = \frac{1}{m_{b}}\sum_{j = 1}^{m_{b}}{\sum_{i = 1}^{n_{jb}}w_{ijb}}\left( \bm{x}_{ijb} - {\overline{\bm{x}}}_{jb} \right)^{\prime}\left( \bm{x}_{ijb} - {\overline{\bm{x}}}_{jb} \right)\).
The intraclass correlation coefficient (ICC) matrix,
\(\bm{\mathit{\Gamma}}_{\bm{x}}\), is defined as
\(\bm{\mathit{\Gamma}}_{\bm{x}} = {(\bm{S}_{\bm{x},\bm{B}}^{2} + \bm{S}_{\bm{x},\bm{W}}^{2})}^{- 1}\bm{S}_{\bm{x},\bm{B}}^{2}\),
where
\(\bm{S}_{\bm{x},\bm{B}}^{2} = \sum_{b = 1}^{h}{q_{b}\bm{S}_{\bm{x},b,\bm{B}}^{2}}\),
\(\bm{S}_{\bm{x},\bm{W}}^{2} = \sum_{b = 1}^{h}{q_{b}\bm{S}_{\bm{x},b,\bm{W}}^{2}}\),
and \(\bm{I}_{v \times v}\) is the identity matrix.
This ICC matrix measures the proportions of the total variances and covariances of the covariates that are due to between-cluster variation.

Next note that \(S_{D_{b}}^{2}\left( 1 \right)\) and \(S_{D_{b}}^{2}\left( 0 \right)\ \)in (8) in the main text are based on cluster-level residuals.
Thus, the within-cluster covariates, (\(\bm{x}_{ij} - {\overline{\bm{x}}}_{j}\)), have no effect on \(S_{D_{b}}^{2}\left( t \right)\), and hence, on precision.
Accordingly, \(\bm{\gamma}\) as defined using the individual data is inefficient.
Intuitively, the within-cluster covariates ``contaminate'' \(\bm{\gamma}\). Comparing \(\bm{\gamma}\) to the optimal \(\bm{\gamma}_{B}\), the difference is \((\bm{I}_{v \times v} - \bm{\mathit{\Gamma}}_{\bm{x}})(\bm{\gamma}_{\bm{W}} - \bm{\gamma}_{\bm{B}})\).
Thus, asymptotic efficiency losses using the individual data will depend on \(\bm{\mathit{\Gamma}}_{\bm{x}}\) and the relative values of \(\bm{\gamma}_{\bm{B}}\) and \(\bm{\gamma}_{\bm{W}}\).
For instance, if the model only contains cluster-level covariates, then \(\bm{\mathit{\Gamma}}_{\bm{x}} = \bm{I}_{v \times v}\) and there are no efficiency losses.
However, if the covariates vary both within and between clusters, \(\bm{\gamma}\) is inefficient.

\subsection{TX collinearity ($R^2$) results using the individual and aggregate data}
We next argue for the following two-part result:
\\
\\
\textbf{Result:}

Consider a non-blocked, clustered RCT design (where we consider a
non-blocked design to reduce notation). Further let the covariate
matrix, \(\bm{X},\) have full rank. Then we have:

\begin{enumerate}
\def\labelenumi{\arabic{enumi}.}
\item
  For any treatment allocation, the \(R_{TXB}^{2}\) value from a
  WLS regression of \({\widetilde{T}}_{j}\) on
  \({\widetilde{\overline{\bm{x}}}}_{j} = ({\overline{\bm{x}}}_{j} - \overline{\overline{\bm{x}}})\)
  using the aggregate data will be at least as large as the
  \(R_{TX}^{2}\) value from a WLS regression of
  \({\widetilde{T}}_{j}\) on \({\widetilde{\bm{x}}}_{ij}\)
  using the individual data.
\item
  If we assume equal cluster sizes (\(\frac{n}{m}\)) and weights of 1,
  we can approximate \(E{(R}_{TXB}^{2})\) using \(\frac{v}{m}\)
  and \(E{(R}_{TX}^{2})\) using
  \(\frac{tr(\bm{\mathit{\Gamma}}_{\bm{x}})}{m} + \frac{(v - tr\left( \bm{\mathit{\Gamma}}_{\bm{x}} \right))}{n}\),
  where \(tr\) is the trace operator and
  \(\bm{\mathit{\Gamma}}_{\bm{x}}\) is a \(v\times v\) matrix of
  intraclass correlation coefficients (ICCs) for the covariates (defined
  below). Here, the expectation is taken over the randomization
  distribution.
\end{enumerate}

\begin{proof}
 To establish that the aggregate \(R_{TXB}^{2}\) is is larger than the individual \(R_{TX}^{2}\) for a given treatment
allocation, we first define \({\widehat{\bm{\lambda}}}_{B}\) to be
the estimated parameter from a WLS regression of \({\widetilde{T}}_{j}\)
on
\({\widetilde{\overline{\bm{x}}}}_{j} = ({\overline{\bm{x}}}_{j} - \overline{\overline{\bm{x}}})\)
using the aggregate data. We have that
\({\widehat{\bm{\lambda}}}_{B} = ({\widetilde{\overline{\bm{X}}}}^{\prime}\bar{\bm{W}} \widetilde{\overline{\bm{X}}})^{- 1}{\widetilde{\overline{\bm{X}}}}^{\prime}\bar{\bm{W}}\widetilde{\bm{T}}\),
where \(\widetilde{\overline{\bm{X}}}\) is an \(m \times v\)
matrix of \({\widetilde{\overline{\bm{x}}}}_{j}\) values,
\(\bar{\bm{W}}\) is an \(m \times m\) diagonal matrix of
\(w_{j}\) weights, and \(\widetilde{\bm{T}}\) is a vector of
\({\widetilde{T}}_{j} = T_{j} - p^{*}\) values, with
\(p^{*}\  = \frac{1}{\sum_{j = 1}^{m}w_{j}}\sum_{j = 1}^{m}{T_{j}w_{j}}\).
Next, using the individual data, consider a WLS regression of
\({\widetilde{T}}_{j}\) on between- and within-cluster covariates,
\(({\overline{\bm{x}}}_{j} - \overline{\overline{\bm{x}}})\)
and
\(\left( \bm{x}_{ij} - {\overline{\bm{x}}}_{j} \right)\),
with associated parameter vectors, \(\bm{\pi}_{B}\) and
\(\bm{\pi}_{W}\). We have
\({\widehat{\bm{\pi}}}_{W}= \bm{0}\) because treatment
assignments are at the cluster level, and thus,
\({\widehat{\bm{\pi}}}_{B} = {\widehat{\bm{\lambda}}}_{B}\).

If we now instead use the individual data to regress
\({\widetilde{T}}_{j}\) on a single set of covariates,
\({\widetilde{\bm{x}}}_{ij}\), with associated parameter
vector, \(\bm{\pi}\), standard decomposition results \citep[see, e.g.,][]{greene_text} establish that
\(\widehat{\bm{\pi}} = \bm{\mathit{\Gamma}}_{\bm{x}}{\widehat{\bm{\lambda}}}_{B}+ (\bm{I}_{v \times v}- \bm{\mathit{\Gamma}}_{\bm{x}}){\widehat{\bm{\lambda}}}_{W} = \bm{\mathit{\Gamma}}_{\bm{x}}{\widehat{\bm{\lambda}}}_{B}\),
where
\(\bm{\mathit{\Gamma}}_{\bm{x}} = ({\widetilde{\bm{X}}}^{\prime}\bm{W}\widetilde{\bm{X}})^{- 1}({\widetilde{\overline{\bm{X}}}}^{\prime}\bar{\bm{W}}\widetilde{\overline{\bm{X}}})\) is
a \(v \times v\) matrix of ICCs for the covariates (that measures the
proportions of the total variances and covariances of the covariates that
are between clusters), \(\widetilde{\bm{X}}\) is an \( n \times v\)
matrix of \({\widetilde{\bm{x}}}_{j}\) values,
\(\bm{W}\) is an \(n \times n\) diagonal matrix of \(w_{ij}\)
weights, and \(\bm{I}_{v \times v}\) is the identity
matrix.

Using these results and noting that the regression \(R^{2}\) value is
the explained sum of squares divided by the total sum of squares, we
have that
\(R_{TXB}^{2} = \frac{1}{np^{*}\left( 1{- p}^{*} \right)}(\widetilde{\bm{T}}^{\prime}\bar{\bm{W}}\widetilde{\overline{\bm{X}}})({\widetilde{\overline{\bm{X}}}}^{\prime}\bar{\bm{W}}\widetilde{\overline{\bm{X}}})^{- 1}({\widetilde{\overline{\bm{X}}}}^{\prime}\bar{\bm{W}}\widetilde{\bm{T}})\),
and similarly for \(R_{TX}^{2}\), which yields

\[R_{TXB}^{2} - R_{TX}^{2} = \frac{1}{np^{*}(1{- p}^{*})}({\widetilde{\bm{T}}}^{\prime}\bar{\bm{W}}\widetilde{\overline{\bm{X}}})(\bm{I}_{v \times v} - \bm{\mathit{\Gamma}}_{\bm{x}})({\widetilde{\overline{\bm{X}}}}^{\prime}\bar{\bm{W}}\widetilde{\overline{\bm{X}}})^{- 1}({\widetilde{\overline{\bm{X}}}}^{\prime}\bar{\bm{W}}\widetilde{\bm{T}}).\]

This result establishes that
\({(R}_{TXB}^{2} - R_{TX}^{2}) \geq 0\) because
\((\bm{I}_{v \times v} - \bm{\mathit{\Gamma}}_{\bm{x}})({\widetilde{\overline{\bm{X}}}}^{\prime}\bar{\bm{W}}\widetilde{\overline{\bm{X}}})^{- 1}\)
is positive semi-definite and
\(({\widetilde{\bm{T}}}^{\prime}\bar{\bm{W}}\widetilde{\overline{\bm{X}}})\)
has full rank, with equality if and only if the covariates do not vary
within clusters (in which case
\(\bm{\mathit{\Gamma}}_{\bm{x}} = \bm{I}_{v \times v}\)).

To next establish the approximation results for
\(E\left( R_{TXB}^{2} \right)\) and
\(E\left( R_{TX}^{2} \right)\), we assume equal cluster sizes
(\(\frac{n}{m}\)) and weights of 1. In this case, for the aggregate
data, we have that
\(E\left( R_{TXB}^{2} \right) = 1 - \frac{1}{mp(1 - p)}E({\widetilde{\bm{T}}}^{\prime}\left( \bm{I} - \bm{P}_{\widetilde{\overline{\bm{X}}}} \right)\widetilde{\bm{T}})\),
where
\(\bm{P}_{\widetilde{\overline{\bm{X}}}} = {\widetilde{\overline{\bm{X}}}}^{\prime}\left( {\widetilde{\overline{\bm{X}}}}^{\prime}\widetilde{\overline{\bm{X}}} \right)^{- 1}{\widetilde{\overline{\bm{X}}}}^{\prime}\)
is the projection matrix. Using the trace operator, we have that
\(E\left( R_{TXB}^{2} \right) = 1 - \frac{1}{mp\left( 1 - p \right)}tr\left\lbrack \left( \bm{I} - \bm{P}_{\widetilde{\overline{\bm{X}}}} \right)\bm{\Lambda} \right\rbrack\), where \(\bm{\ \Lambda} =E\left( \widetilde{\bm{T}}{\widetilde{\bm{T}}}^{\prime} \right)\)
has diagonal elements, \(p(1 - p)\), and off-diagonal covariances,
\(- \frac{p(1 - p)}{(m - 1)}\). If we ignore the small covariances in
\(\bm{\Lambda}\), we find that
\(E\left( R_{TXB}^{2} \right) \approx \frac{v}{m}\). Similarly,
for the individual data, we have that
\(E\left( R_{TX}^{2} \right) = 1 - \frac{1}{mp(1 - p)}E({\widetilde{\bm{T}}}^{\prime}\left( \bm{I} -\bm{P}_{\widetilde{\bm{X}}} \right)\widetilde{\bm{T}})\),
where
\(\bm{P}_{\widetilde{\bm{X}}} = {\widetilde{\overline{\bm{X}}}}^{\prime}\bm{\mathit{\Gamma}}_{\bm{x}}\left( {\widetilde{\overline{\bm{X}}}}^{\prime}\widetilde{\overline{\bm{X}}} \right)^{- 1}{\widetilde{\overline{\bm{X}}}}^{\prime}\)
is the new projection matrix. If we again use the trace operator and
ignore the off-diagonal elements in \(\bm{\Lambda}\), we find that
\(E\left( R_{TX}^{2} \right) \approx \frac{tr(\bm{\mathit{\Gamma}}_{\bm{x}})}{m}\).
To avoid a zero expected value when
\(\text{tr}\left( \bm{\mathit{\Gamma}}_{\bm{x}} \right) = 0\), we instead
use an alternative approximation,
\(E{(R}_{TX}^{2}) \approx \ \frac{tr\left( \bm{\mathit{\Gamma}}_{\bm{x}} \right)}{m} + \frac{(v - tr\left( \bm{\mathit{\Gamma}}_{\bm{x}} \right))}{n}\).
\end{proof}

To motivate this approximation for \(E{(R}_{TX}^{2})\), we
express it as \(E{(R}_{TX}^{2}) \approx \ \frac{v}{n^{*}}\),
where
\(n^{*} = \frac{n}{1 + {\overline{\rho}}_{X}\left( \overline{n} - 1 \right)}\)
is the ``effective'' sample size of individuals,
\({\overline{\rho}}_{X} =\)
\(\frac{tr\left( \bm{\mathit{\Gamma}}_{\bm{x}} \right)}{v}\) is
the average covariate ICC for the variances, and
\(\overline{n} = \frac{n}{m}\) is the average cluster size. The
denominator term,
\(\lbrack 1 + {\overline{\rho}}_{X}\left( \overline{n} - 1 \right)\rbrack\),
can be considered a design effect due to covariate clustering. As
\({\overline{\rho}}_{X}\) approaches 0, \(n^{*}\) approaches \(n\)
(minimum design effects), whereas as \({\overline{\rho}}_{X}\)
approaches 1, \(n^{*}\) approaches \(m\) (maximum design effects).

Table A.1 displays simulation results that support the use of these
approximations. For the simulations, we randomly generated
\(v = 2,\ 5,\) or \(10\) covariates using \(x_{ij1} = u_{j1} + e_{ij1}\)
and
\(x_{ijk} = \theta x_{ij(k - 1)} + u_{jk} + e_{ijk}\)
for \(k = 2,\ldots,v\), where \(u_{jq}\) and \(e_{jq}\)
are independently and identically distributed normal random errors with
mean zero. The within-cluster errors, \(e_{jq}\), were set to
have variance 1, and the variances of the between-cluster errors,
\(u_{jq}\), were calculated based on assumed ICC values of
\(\rho_{X} = 0,\ 0.4,\) or 0\(.8\). The parameter, \(\theta\), was
calculated to generate a 0\(.5\) correlation coefficient (\(r\)) between
\(x_{ijk}\) and \(x_{ij(k - 1)}\). Specifically, we set
\(\theta = \frac{r}{\sqrt{1 - r^{2}}} = .577\). The number of clusters
ranged from \(m = 20\) to \(60\), and to allow for some unbalance in the
design, we set \(p = .6\) and allowed the number of individuals per
cluster to range uniformly between \(n_{j} = 25\) and \(75\). For each
model specification, we generated a single dataset of covariates and
then generated 500 draws from the randomization distribution. We
repeated this procedure 10 times to avoid unusual base datasets and
calculated average results. Table A.1 presents mean
\(R_{TXB}^{2}\) and \(R_{TX}^{2}\) values across the 5,000
simulations using the individual and aggregate data (the ``true''
values) as well as the \(E(R_{TXB}^{2})\) and
\(E(R_{TX}^{2})\) approximations from above.

The results indicate that the \(E(R_{TXB}^{2})\) approximations
using the aggregate data are close to true values, with slight downward
biases as \(\rho_{X}\) values increase. The \(E(R_{TX}^{2})\)
approximations are also reasonably close to true values, with no
patterns of downward or upward biases.

\includepdf[pages=1-last]{"suppAtable".pdf}

\subsection{Details Comparing Design-Based and CRSE Estimators}
Using individual-level data, the CRSE variance estimator for the WLS coefficients for model  (5) in the main text that includes baseline covariates is:
\begin{align}
\widehat{\text{Var}}_{\text{CRSE}} (\hat{\bm{\delta}} )=g(\bm{Z}^{\prime} \bm{W}\bm{Z})^{-1} \left(\sum_{b=1}^h\sum_{j=1}^{m_b}\bm{Z}_{jb}^{\prime}  \bm{W}_{jb} \hat{\bm{e}}_{jb} \hat{\bm{e}}_{jb}^{\prime}  \bm{W}_{jb}  \bm{Z}_{jb}\right) (\bm{Z}^{\prime}  \bm{W}\bm{Z})^{-1}, \label{eq:crse_est}
\end{align}
where $\bm{Z}$ is an $n\times(2h+v)$ matrix with the full set of independent variables (the block-by-treatment status interactions, block indicators, and covariates); $\hat{\bm{\delta}}$ is the corresponding vector of coefficient estimates; $\bm{W}$ is an $n\times n$ symmetric weight matrix with diagonal entries $w_{ijb}$ and all other entries 0; $\bm{W}_{jb}$ are $n_{jb} \times n_{jb}$ submatrices of $\bm{W}$ defined for each cluster; $\hat{\bm{e}}_{jb}$ is a vector of WLS residuals for each cluster; and $g$ is a small sample correction term discussed in Section 4.3 in the main text. 

Consider a single block-by-treatment ATE estimator from the regression model. In this case, after some algebra, the CRSE variance estimator in (\ref{eq:crse_est}) reduces to
\begin{align*}
\widehat{\text{Var}}_{\text{CRSE}}(\hat{\beta}_{1,b} )=g \frac{s_{D_b}^{2*} (1)}{m_b^1}+g \frac{s_{D_b}^{2*} (0)}{m_b^0 },	
\end{align*}
where $s_{D_b}^{2*} (1)=\frac{m_b^1-1}{m_b^1} s_{D_b}^2 (1)$ and $s_{D_b}^{2*} (0)=\frac{m_b^0-1}{m_b^0} s_{D_b}^2 (0)$.
This establishes (10) in the main text. 

Next consider adding covariates to the model. In this case, using (\ref{eq:crse_est}), the CRSE variance estimator for $\hat{\beta}_{1,b} $ becomes
\begin{align*}
\widehat{\text{Var}}_{\text{CRSE}}(\hat{\beta}_{1,b} )=g \frac{\sum_{s=1}^h\sum_{j=1}^{m_s}\hat{\bm{\xi}}^{\prime}_{js,b}\bm{W}_{js}\hat{\bm{e}}_{js}\hat{\bm{e}}_{jb}^{\prime}\bm{W}_{js}\hat{\bm{\xi}}_{js,b}}{\left(\sum_{s=1}^h\sum_{j=1}^{m_s}\hat{\bm{\xi}}^{\prime}_{js,b}\bm{W}_{js}\hat{\bm{\xi}}_{js,b}\right)^2 },
\end{align*}
where $\hat{\bm{\xi}}_{js,b}$ is the residual in block $s$ from a weighted regression of $S_{ijb} \tilde{T}_{jb}$ on the other block-by-treatment status interactions and covariates in $\bm{Z}$.
In comparison, the design-based estimator in (9) of the main text can be expressed in a parallel matrix form as

\begin{align*}
\widehat{\text{Var}}(\hat{D}_b )=&\frac{m_b^1}{(m_b^1-v^* p_b^* q_b^*-1) }\frac{\sum_{s=1}^h\sum_{j: T_{js}=1}^{m_s}\hat{\bm{\eta}}^{\prime}_{js,b}\bm{W}_{js}\hat{\bm{e}}_{js}\hat{\bm{e}}_{jb}^{\prime}\bm{W}_{js}\hat{\bm{\eta}}_{js,b}}{\left(\sum_{s=1}^h\sum_{j: T_{js}=1}^{m_s}\hat{\bm{\eta}}^{\prime}_{js,b}\bm{W}_{js}\hat{\bm{\eta}}_{js,b}\right)^2}\\
& + \frac{m_b^0}{(m_b^0-v^* (1-p_b^*) q_b^*-1) }\frac{\sum_{s=1}^h\sum_{j: T_{js}=0}^{m_s}\hat{\bm{\eta}}^{\prime}_{js,b}\bm{W}_{js}\hat{\bm{e}}_{js}\hat{\bm{e}}_{jb}^{\prime}\bm{W}_{js}\hat{\bm{\eta}}_{js,b}}{\left(\sum_{s=1}^h\sum_{j: T_{js}=0}^{m_s}\hat{\bm{\eta}}^{\prime}_{js,b}\bm{W}_{js}\hat{\bm{\eta}}_{js,b}\right)^2},
\end{align*}
where $\hat{\bm{\eta}}_{js,b}$ is now the residual in block s from a weighted regression of $S_{ijb} \tilde{T}_{jb}$ on the other  block-by-treatment status interactions, but \textit{not} the covariates.
As with the no covariate case, the two estimators differ in their degrees of freedom adjustments.
However, they now also differ in the handling of the correlations between the treatment indicators and covariates ($\hat{\bm{\xi}}_{js,b}$ versus $\hat{\bm{\eta}}_{js,b}$).
The two variance estimators, however, are again asymptotically equivalent.

\includepdf[pages=1-last]{"appendBtext".pdf}
\includepdf[landscape=true, pages=1-last]{"tableB1".pdf}
\includepdf[pages=1-last]{"tableB2".pdf}
\includepdf[landscape=true, pages=1-last]{"tableB3".pdf}
\includepdf[pages=1-last]{"tableB45".pdf}
\includepdf[landscape=true, pages=1-last]{"tableB6".pdf}


\bibliographystyle{apalike}
\bibliography{clt_ref}{}
\end{document}